\documentclass[pdflatex,sn-mathphys-num]{sn-jnl}

\usepackage{graphicx}%
\usepackage{multirow}%
\usepackage{amsmath,amssymb,amsfonts}%
\usepackage{amsthm}%
\usepackage{mathrsfs}%
\usepackage[title]{appendix}%
\usepackage{xcolor}%
\usepackage{textcomp}%
\usepackage{manyfoot}%
\usepackage{booktabs}%
\usepackage{algorithm}%
\usepackage{algorithmicx}%
\usepackage{algpseudocode}%
\usepackage{listings}%
\usepackage{natbib}
\usepackage{hyperref}
\usepackage{lscape}
\usepackage{color}
\usepackage{colortbl}
\usepackage{setspace}
\usepackage{dsfont}
\usepackage[hang,flushmargin]{footmisc} 
\usepackage{graphics}
\usepackage{subcaption}
\usepackage{stmaryrd}
\usepackage{bbm}
\usepackage{enumitem}
\usepackage{sidecap}
\usepackage{xparse}
\usepackage{sansmath}

\usepackage{mathtools}
\mathtoolsset{showonlyrefs} 

\newtheorem{theorem}{Theorem}[section]
\newtheorem{lemma}[theorem]{Lemma}

\newtheorem{proposition}[theorem]{Proposition}

\newtheorem{remark}[theorem]{Remark}

\makeatletter 
\@addtoreset{equation}{section}
\def\theequation{\arabic{section}.\arabic{equation}}
\makeatother 
\newcommand{\prob}{\mathsf{P}}

\newcommand{\media}{\mathsf{E}}
\newcommand{\beq}{\begin{equation}}
\newcommand{\eeq}{\end{equation}}

\newcommand{\reali}{\ensuremath{\mathbb{R}}}
\newcommand{\naturali}{\ensuremath{\mathbb{N}}}

\newcommand{\cH}{\mathcal{H}}

\newcommand{\cF}{\mathcal{F}}
\newcommand{\cG}{\mathcal{G}}

\newcommand{\ind}{\boldsymbol{1}}
\newcommand{\bG}{\ensuremath{\mathbb{G}}}

\newcommand{\bH}{\ensuremath{\mathbb{H}}}
\newcommand{\bF}{\ensuremath{\mathbb{F}}}

\def \ud{\hspace{1pt} \mathrm{d}}
\def \e{\mathrm{e}}
\def \given{\hspace{1.5pt}|\hspace{1.5pt}}
\def \Given{\hspace{1.5pt}\Big|\hspace{1.5pt}}

\begin{document}

\title{\textbf{Filtering in a hazard rate change-point model with financial and life-insurance applications}}

\author*[1]{\fnm{Matteo} \sur{Buttarazzi}}\email{matteo.buttarazzi@unito.it}
\author[2]{\fnm{Claudia} \sur{Ceci}}\email{claudia.ceci@uniroma1.it}

\affil[1]{{\orgdiv School of Management and Economics, Dept. ESOMAS}, \orgname{University of Torino}, \orgaddress{\street{Corso
Unione Sovietica, 218 Bis, 10134, Torino, Italy}}}

\affil[2]{\orgdiv{Department MEMOTEF}, \orgname{Sapienza University of Rome}, \orgaddress{\street{Via del Castro Laurenziano, 9}, \city{Rome}, Italy}}
\abstract{This paper develops a continuous-time filtering framework for estimating a hazard rate subject to an unobservable change-point. This framework naturally arises in both financial and insurance applications, where the default intensity of a firm or the mortality rate of an individual may experience a sudden jump at an unobservable time, representing, for instance, a shift in the firm’s risk profile or a deterioration in an individual’s health status.
By employing a progressive enlargement of filtration, we integrate noisy observations of the hazard rate with default-related information. We characterise the filter, i.e. the conditional probability of the change-point given the information flow, as the unique strong solution to a stochastic differential equation driven by the innovation process enriched with the discontinuous component. 
A sensitivity analysis and a comparison of the filter's behaviour under various information structures are provided.
Our framework further allows for the derivation of an explicit formula for the survival probability conditional on partial information. This result applies to the pricing of credit-sensitive financial instruments such as defaultable bonds, credit default swaps, and life insurance contracts.
Finally, a numerical analysis illustrates how partial information leads to delayed adjustments in the estimation of the hazard rate and consequently to mispricing of credit-sensitive instruments when compared to a full-information setting.}

\maketitle

\section{Introduction}
In many real-world applications, one often only has access to partial, limited, or noisy data, making it challenging to derive meaningful information. This is where stochastic filtering comes into play. Filtering problems concern the estimation of an unobserved stochastic process $(X_t)_{t\geq0}$, referred to as the signal,  given observations of a related process $(Y_t)_{t\geq0}$. This leads to computing, for each time $t$, the conditional distribution of $X_t$ given the information flow $\mathcal{F}^Y_t=\sigma\{Y_s, s\leq t\}$, which provides the best estimate of the signal according to the mean-square error.
This problem was solved for linear Gaussian systems by \cite{kalman1}, \cite{kalman2}, leading to the widely used Kalman filter. However, many practical systems exhibit nonlinear dynamics and non-Gaussian noise, rendering the filtering problem significantly more complex and generally infinite-dimensional.
Early contributions from Kushner \cite{kushner}, Stratonovich \cite{stratonovich}, Kailath \cite{kailath} and Zakai \cite{zakai} laid the foundation for nonlinear filtering. 
Building on these foundational works, filtering theory has continued to evolve, extending to more complex settings, including those involving jump processes. In particular, filtering for pure-jump or jump-diffusion models has been studied extensively, with significant contributions from \cite{Bain-Crisan}, \cite{bremaud1981point}, \cite{ceci2012nonlinear}, \cite{ceciG2000filtering}, \cite{ceciG2001nonlinear}  and \cite{frey2012pricing}, among others. In these discontinuous settings, the innovation process is enriched with a compensated jump component.

A closely related framework is provided by the detection problem, which aims to determine whether a certain hypothesis about an observed process is true or not. 
In its Bayesian formulation, detection involves sequentially testing between statistical hypotheses regarding the underlying probability measure governing an observed process.
Classical examples include testing hypotheses about the mean of a Wiener process or the intensity of a Poisson process
(cf. \cite{61shiryaev}, \cite{63shiryaev}, \cite[Ch. 6]{peskir2006optimal} and \cite[Ch. 4]{shiryaev2007optimal}).
See also \cite{Bayraktar}, \cite{deangelis2022quickestdetectionproblemfalse}, \cite{Ernst} and \cite{Gapeev} for recent developments in this area.
Filtering and detection are closely related: solving a detection problem may involve estimating processes under different probability measures, while filtering problems can sometimes be approached by introducing auxiliary probability measures, as discussed in \cite{Wong1985}.

The study of partially observable hazard rates has been explored in  \cite{CCC1}, \cite{CCC2}, \cite{CCC3}, \cite{frey2012pricing}, among others. In both \cite{frey2012pricing} and \cite{CCC1}, the residual lifetimes of a set of individuals are modelled as conditionally independent, doubly stochastic random times. 
Analogously to our setting, \cite{frey2012pricing} considers a mixed-type information structure incorporating default events and noisy observations of the Markov chain.
The authors derive traded security prices and compute risk-minimizing hedging strategies by applying the innovation approach. In \cite{CCC1}, the authors analyse a local risk-minimization approach in a combined financial-insurance framework. The insurer has complete information about the financial market and the number of surviving policyholders, but cannot observe the underlying hazard rate. This framework is applied to hedge unit-linked life insurance contracts under partial information.

The filtering setting in \cite{CCC2} and \cite{CCC3} involves enlargement of filtrations but under different information structures: continuous observations in \cite{CCC2} and pure discontinuous observations given by the knowledge of the death status of the policyholder in \cite{CCC3}. 
Whereas \cite{CCC2} focuses on hedging unit-linked life insurance contracts under partial information, \cite{CCC3} addresses pricing pure endowment contracts.

In this paper, we study the filtering problem for a hazard rate change-point model. Specifically, we consider a Bayesian setting where the hazard rate associated with an exogenous random time $\tau$, which may represent a firm's default time or a policyholder's time of death, evolves according to a \textit{single jump process} $(\mu_t)_{t\ge0}$. This process exhibits a jump at an unobservable time $\xi$, representing the moment when an event occurs that alters the firm's risk profile or an individual's health status.
The random time $\xi$ takes the value $0$ with probability $\pi$, and conditional on $\xi>0$ it is exponentially distributed with parameter $\lambda>0$.
The goal is to estimate the hazard rate process $(\mu_t)_{t\ge0}$ given the available information. Due to the structure of $(\mu_t)_{t\ge0}$, this reduces to estimating the conditional probability $(\Pi_t)_{t\ge 0}$ of the jump time $\xi$ given the information flow.
Our approach builds on \cite[pp. 308-310]{peskir2006optimal}, where a similar problem is studied. Although they estimate the time $\xi$ based only on noisy observations of $(\mu_t)_{t\ge0}$, we assume that the default or survival status is also observable, leading to mixed-type observations. 
We model the random time $\tau$ as a doubly stochastic random time and employ a progressive enlargement of the filtration to integrate the information on its occurrence into our analysis. By extending the change-of-measure techniques of \cite[pp. 308-310]{peskir2006optimal} to this progressively enlarged framework, we characterize the process $(\Pi_t)_{t\ge 0}$ as the unique strong solution to a stochastic differential equation driven by the \textit{innovation process} enriched with the discontinuous component.
To the best of our knowledge, this is the first time that this approach has been discussed in a progressively enlarged framework.
The second main contribution of this paper lies in the derivation, via a PDE approach, of closed-form expressions for both the conditional survival probability and the conditional density of the random time $\tau$. These analytical results form the foundation for practical applications in pricing credit/mortality-sensitive instruments under partial information. 
We emphasize that we derive the filtering equation starting from a general underlying filtration $\bF$, without requiring that the hazard-rate process $(\mu_t)_{t\ge 0}$ is Markov with respect to $\bF$. Therefore, we cannot apply classical filtering techniques based on the innovation approach, as in \cite{frey2012pricing}, where a similar framework is considered under the assumption that the unobservable process follows a finite-state Markov chain. Moreover, our derivation is developed within a Bayesian framework. The Bayesian setting also directly connects to quickest detection problems, where one seeks to identify the change-point as rapidly and accurately as possible. Consequently, our results may be relevant to a wide range of applications in finance and actuarial science involving a hidden regime switch that affects system dynamics.

The paper is organised as follows.
\begin{itemize}
    \item Section \ref{Sec:modelingfr} formulates the nonlinear filtering problem for a hazard rate change-point model with mixed-type observations: noisy observation of the hazard rate and default/death-related information. 

    \item Section \ref{Sec:Filter} solves the filtering problem via a change-of-measure approach, extending \cite[pp. 308–310]{peskir2006optimal} to a progressively enlarged filtration setting. The filter $(\Pi_t)_{t \ge 0}$ is characterised as the unique strong solution to the filtering equation. We further analyse parameter sensitivity and compare filter behaviour across different information flows.
    
    \item Section \ref{Sec:finapp} applies the filtering results to price credit/mortality-sensitive instruments under partial information. Section \ref{Sec:condprob} derives closed-form expressions for the conditional survival probability and the conditional density of the random time $\tau$. Section \ref{Sec:Creditd} obtains explicit pricing formulas for credit derivatives. In the numerical analysis, we compare prices under partial and full information. Section \ref{Sec:extension} extends the framework to price, under restricted information, instruments contingent on both default/mortality events and exogenous market factors, assuming independence between them.
    
    \item[] Finally, some technical proofs and auxiliary results are collected in Appendix \ref{App:proofs}.
\end{itemize}

\section{Modeling framework} \label{Sec:modelingfr}
We consider a Bayesian framework to model the change in the hazard rate $(\mu_t)_{t\ge0}$ of a firm's default time $\tau$, which undergoes a shift at an unobservable random time $\xi$.
In this setting, we have partial information about  $(\mu_t)_{t\ge0}$ consisting of the knowledge, at any time $t$, of whether the default event has occurred or not, as well as noisy observations of  $(\mu_t)_{t\ge0}$ through a Gaussian additive noise model.
Thus, the available information is modelled by the progressive enlargement of the filtration generated by the observation of a diffusion process $(\hat Y_t)_{t\ge0}$ with the filtration generated by the default indicator process $(H_t)_{t\ge 0}$.

To keep in line with the literature on quickest detection problems, we consider a filtered probability space $(\Omega,\cF, \bF, \prob_\pi)$ where the filtration $\bF = (\cF_t)_{t\ge 0}$ satisfies the usual hypothesis of right-continuity and completeness. The probability measure $\prob_\pi$ has the following structure
\begin{equation} \label{definitionProbPi}
\prob_\pi = \pi \prob^0+(1-\pi)\int_0^\infty \lambda \e^{-\lambda s}\prob^s \ud s, \end{equation}
for $\pi\in[0,1]$ and where $\prob^s$ is the probability measure under which the default time $\tau$ in \eqref{construction} has a hazard rate that shifts from $\mu_\ell>0$ to $\mu_h>0$ at time $s \in [0,+\infty)$   
and the observed process $(\hat Y_t)_{t\ge0}$ undergoes a drift change at the same time $s$ (cf. \eqref{Addnoise}).
Let $\xi$ be a non-negative random variable such that, for $\lambda>0$, \begin{align}
\prob_\pi(\xi=0)=\pi \quad \text{and} \quad \prob_\pi(\xi>t\given \xi>0)=\e^{-\lambda t}, \quad t\geq 0.    
\end{align} 
In our setting, for any $A\in\cF$,  we have that \begin{equation}\label{definitionProbs} \prob^s(A) =\prob_\pi(A\given \xi=s), \quad s\geq 0. \end{equation}

We introduce an $\bF$-adapted \textit{hazard rate} process $(\mu_t)_{t\ge0}$ that describes the dynamic evolution of a firm's default risk over time.
Let $\Delta\mu \coloneqq \mu_h-\mu_\ell$ and define 
\begin{align}
    \mu_t \coloneqq  \mu_\ell + \Delta\mu \ind_{\{t\geq \xi\}}= \mu_\ell \ind_{\{t<\xi\}} + \mu_h \ind_{\{t\geq \xi\}}, \quad t\ge0.
\end{align}
Initially, the firm's hazard rate is determined by the random variable $\mu_0= \mu_\ell \ind_{\{\xi>0\}} + \mu_h \ind_{\{\xi=0\}}$, so that the sigma-field $\cF_0$ is non-trivial as $\sigma\{\mu_0\}\subseteq \cF_0$.
The random variable $\xi$ represents the time at which an event occurs that alters the firm's risk profile. Specifically, when $\xi>0$, the firm's initial hazard rate is $\mu_\ell$, and at time $\xi$ the event takes place, shifting the hazard rate to $\mu_h$. When $\xi=0$, the change in the hazard rate occurs at time $t=0$ and $\mu_t=\mu_h$, for any $t\geq 0$. Notice that since $(\mu_t)_{t\ge 0}$ is $\bF$-adapted, $\xi$ is an $\bF$-stopping time. 
We assume that $\xi$ is not $\cF_0$-measurable so that the exact occurrence time of the change is not known at the initial time $0$, making the filtering problem nontrivial. Notice that $(\mu_t)_{t\ge 0}$ may not be an $(\bF, \prob_\pi)$-Markov process.

We now model the default time of the firm $\tau$ as an exogenous doubly stochastic random time with $\bF$-hazard rate $(\mu_t)_{t \geq 0}$ (cf. \cite[Ch. 8.2.1]{bielecki2013credit}, \cite[Ch. 2.3]{aksamit2017enlargement}).
According to the canonical construction, we assume that there exists, on the space $(\Omega,\cF, \prob_\pi)$, a random variable $\Theta$ independent of $\cF_\infty \coloneqq \vee_{t\geq0}\cF_t$ and exponentially distributed with parameter one. 
We define $\tau$ as the first time when the strictly increasing process $\Lambda_t \coloneqq \int_0^t \mu_s \ud s$ is above the random level $\Theta$, that is
\begin{equation}\label{construction}
\tau \coloneqq \inf \{t\geq 0: \Lambda_t\geq \Theta\}. 
\end{equation}
It is well-known that (\cite[Lemma 7.3.2.1]{jeanblanc2009mathematical}) 
\begin{align}\begin{aligned} \label{Eq:tau_d_F_infty} \prob_\pi(\tau>s|\cF_s)=\prob_\pi(\tau>s|\cF_\infty)=\exp(-\Lambda_s),\quad \text{for all $s\ge 0$}. 
\end{aligned}\end{align}
\begin{remark} \label{Rm:avoidance}
Since $\tau$ is a finite random time and the Az\'ema supermartingale $Z_t \coloneqq \exp(-\Lambda_t)$ is continuous, it follows that $\tau$ avoids $\bF$-stopping times. Specifically, $\prob_\pi(\tau = \sigma < + \infty)=0$ for any $\bF$-stopping time $\sigma$ (cf. \cite[Proposition 3.3]{coculescu2012hazard}).    
\end{remark}
Let $H_t  \coloneqq  \ind_{\{\tau\leq t\}}$ for $t\ge0$ be the default indicator process associated to $\tau$ and define $$\cH_t  \coloneqq  \sigma\{H_u, 0 \leq u \leq t\}, \quad t\ge 0.$$
Since the random time $\tau$ is not an $\bF$-stopping time, we consider a progressive enlargement of the filtration $\bF$. Let $ {\bG}  \coloneqq  (\cG_t)_{t\geq 0}$ denote the progressively enlarged filtration given by
\begin{equation}
 \cG_t  \coloneqq  \cF_t \vee  \mathcal{H}_t \quad t\ge0.
\end{equation}
In particular, ${\bG}$ is the smallest filtration which contains $\bF$ and such that $\tau$ is a ${\bG}$-stopping time and plays the role of the market full information.

\begin{remark}\label{Rm:immersion} As an immediate consequence of the canonical construction, see \eqref{construction}, we get that
the so-called \textit{Immersion property} between filtrations $\bF$ and $\bG$ holds, i.e. every $\bF$-(local) martingale is also a $\bG$-(local) martingale, see \cite{bremaud1978changes} or \cite{aksamit2017enlargement}. Moreover, the process 
\begin{equation}
H_t- \int_0^{\tau \wedge  t} \mu_s \ud s = H_t- \int_0^t (1-H_s)\mu_s \ud s , \quad t \geq 0,\end{equation}
is a $( \bG,\prob_\pi)$-martingale and $\tau$ is a totally inaccessible $\bG$-stopping time. 
\end{remark}

\begin{remark}
Although we explicitly discuss a firm's default risk, the framework can equally be applied to individual mortality risk. In this interpretation, the random time $\tau$ denotes the time of death, and the process $(\mu_t)_{t\ge0}$ represents the mortality force process that describes the change in an individual's mortality risk profile occurring at a random time $\xi$.
\end{remark}

We suppose that the change-point $\xi$ is unobservable and therefore we have restricted information on the hazard rate $(\mu_t)_{t\geq 0}$.
Similar to \cite{frey2012pricing}, we assume partial information of the type 
\begin{equation} \label{Addnoise} \ud \hat{Y}_t=\mu_t \ud t + \beta \ud B_t, \quad \hat{Y}_0=0,\quad \end{equation}
where $\beta>0$ and $(B_t)_{t\geq0}$ is an $\bF$-adapted standard Brownian motion independent of $\xi$. Let $Y_t \coloneqq \hat{Y}_t-\mu_\ell t$, then
\begin{equation} \label{Eq:dY} \ud Y_t=(\mu_h-\mu_\ell)\ind_{\{t\geq\xi\}}\ud t+\beta \ud B_t, \quad Y_0=0, \end{equation}
hence
\begin{equation} Y_t=\begin{cases}
    \beta B_t &\text{if } t<\xi\\
    \Delta\mu (t-\xi) + \beta B_t &\text{if } t\geq\xi.
\end{cases} \end{equation}
Let $\bF^{Y} \coloneqq ( \cF^Y_t)_{t\geq 0}$ denote the filtration generated by the sample paths of the process $(Y_t)_{t\geq 0}$, i.e. each $\sigma$-field is given by  \begin{equation} \cF^Y_t \coloneqq \sigma\{Y_s,s\leq t\}. \end{equation}
Notice that the filtration generated by the process $(\hat Y_t)_{t\geq 0}$ coincides with $\bF^{Y}$. 

The information flow available to the individual is represented by the progressively enlarged filtration $\bG^Y=(\cG^Y_t)_{t \geq 0}$ defined as
\begin{equation}  \cG^Y_t  \coloneqq  \cF^Y_t \vee  \mathcal{H}_t \subset \cF_t \vee  \mathcal{H}_t = \cG_t, \quad t\geq 0.\end{equation}
The filtration $\bG^Y$ captures two key sources of information. First, it includes the observation of the process $(\hat{Y})_{t \geq 0}$, which provides a Gaussian additive noisy observation of the firm's hazard rate $(\mu_t)_{t \geq 0}$. Second, it incorporates default-related information through the process $(H_t)_{t\geq 0}$, which indicates whether a default has occurred ($H_t = 1$) or not ($H_t = 0$) up to time $t$.

The goal is to obtain the best estimate of $(\mu_t)_{t\geq 0}$ given the available information. According to the filtering literature, this estimate is provided by the process $(\hat \mu_t = \media_\pi[ \mu_t \given  \cG^Y_t])_{t\geq 0}$, which in our setting can be written as 
\begin{equation} \label{G-estimate} \hat \mu_t = \media_\pi[ \mu_t \given  \cG^Y_t] = \mu_\ell (1- \Pi_t) + \mu_h  \Pi_t =  \mu_\ell + \Delta \mu \Pi_t, \quad t \geq 0,\end{equation}
where the filter \begin{equation} \label{DefinitionPi} \Pi_t \coloneqq  \prob_\pi(\xi\leq t \given \cG^Y_t), \quad t \geq 0,\end{equation} provides the conditional distribution of $\xi$ given $\cG^Y_t$, for any $t \geq 0$.
According to \cite[Lemma 1.1]{kurtz1988unique}, there exists a càdlàg version of the processes $(\hat{\mu}_t)_{t \geq 0}$ and $(\Pi_t)_{t \geq 0}$. As usual, for a càdlàg process $(R_t)_{t\ge0}$, we denote by $(R_{t-})_{t\ge0}$ its left-continuous version, that is $R_{t-}  \coloneqq  \lim_{s\to {t-}} R_s$.

\section{The filtering problem} \label{Sec:Filter}
In this section, we characterize the process $(\Pi_t)_{t\ge 0}$ as the unique strong solution to a stochastic differential equation (SDE), referred to as the filtering equation, by extending the change-of-measure techniques of \cite[pp. 308-310]{peskir2006optimal} to an enlarged filtration framework. 
The filtering equation is then employed in Section \ref{Sec:finapp} to compute conditional survival probabilities
and to price credit/mortality-sensitive contracts under partial information.

We first establish key preliminary results in Propositions \ref{Prop:Probabilities}--\ref{App:hatBhatm}, whose proofs are provided in Appendix \ref{App:proofs}. We start by defining $\prob^\infty$ as the probability measure under which the default time $\tau$ has a constant hazard rate equal to $\mu_\ell$ and the process $(Y_t)_{t\ge 0}$ undergoes no change in drift, i.e., $Y_t=\beta B_t$, $t\geq 0$.
\begin{proposition} \label{Prop:Probabilities}
For any $t\ge0$, the following equalities hold $\prob_\pi$-a.s.:
\begin{equation} \label{Probabilities_and_Pi}
\Pi_t=\pi\frac{\ud \prob^0}{\ud \prob_\pi}\Big|_{ \cG^Y_t}+(1-\pi)\int_0^t\frac{\ud \prob^s}{\ud \prob_\pi}\Big|_{ \cG^Y_t}\lambda \e^{-\lambda s}\ud s
\end{equation} 
and
\begin{equation} \label{Probabilities_and_1-Pi}
1-\Pi_t=(1-\pi)\e^{-\lambda t}\frac{\ud \prob^t}{\ud \prob_\pi}\Big|_{ \cG^Y_t}=(1-\pi)\e^{-\lambda t}\frac{\ud \prob^\infty}{\ud \prob_\pi}\Big|_{ \cG^Y_t}.
\end{equation} 
\end{proposition}

It follows from \eqref{Probabilities_and_1-Pi} that, for any $\pi \in [0,1)$, the process $(\Pi_t)_{t \ge 0}$ remains in $[0,1)$ $\prob_\pi$-a.s.
For $\pi\in[0,1)$, we introduce an auxiliary process
\begin{equation} \label{Eq:varphi}
    \varphi_t  \coloneqq  \frac{\Pi_t}{1 - \Pi_t}, \quad t\ge0,
\end{equation}
and, in the following proposition, we derive the stochastic differential equation which it solves.

\begin{proposition}\label{dvarphi} Let $\pi\in[0,1)$, the process $(\varphi_t)_{t\ge 0}$ is solution to
    \begin{align}
    \begin{aligned} \label{dvarphieq}
        \ud \varphi_t & = \lambda( 1 + \varphi_t) \ud t +  \varphi_{t^-} \ud M_t, \quad  \varphi_0  =  \frac{\pi}{1 - \pi}
    \end{aligned}
\end{align}
where $(M_t)_{t\geq 0}$ is the $( \bG^Y,\prob^\infty)$-martingale given by
\begin{align}
\label{dM} \ud M_t \coloneqq  \frac{\Delta \mu}{\beta^2}  \ud Y_t + \frac{\Delta \mu}{\mu_\ell} \big (\ud H_t  -\mu_\ell(1-H_{t^-})  \ud t).    \end{align}

\end{proposition}

We introduce the \textit{Innovation process}, consisting of the pair $\big( (\hat B_t)_{t\ge0},  (\hat m_t)_{t\ge0} \big)$ given by
\begin{equation}\label{innovationB}  \hat  B_t \coloneqq \frac{1}{\beta}\bigg( Y_t - \Delta\mu\int_0^t \Pi_s \ud s \bigg), \quad t \geq 0,\end{equation}
\begin{equation} \label{innovationm} \hat m_t \coloneqq H_t  - \int_0^t (1 - H_{s^-})\hat \mu_{s^-} \ud s , \quad t \geq 0,\end{equation}
with $(\hat \mu_t)_{t\ge 0}$ defined in \eqref{G-estimate}. 

\begin{proposition} \label{App:hatBhatm}
The process $(\hat B_t)_{t\ge0}$ defined in \eqref{innovationB} is a $(\bG^Y,\prob_\pi)$-Brownian motion and $(\hat m_t)_{t\ge0}$, defined in \eqref{innovationm}, is a $(\bG^Y,\prob_\pi)$-martingale. 
\end{proposition}

The following theorem is the main result of the section. We derive the stochastic differential equation of Kushner-Stratonovich type that the filter $(\Pi_t)_{t\geq 0}$ solves.
\begin{theorem} \label{Th:SDEFilter}
For $\pi\in[0,1]$ and under $\prob_\pi$, the process $(\Pi_t)_{t\geq 0}$ solves the following stochastic differential equation    
\begin{align} \label{SDEFilterG} \ud \Pi_t = \lambda(1-\Pi_t)\ud t+\frac{\Delta\mu}{\beta}\Pi_t(1-\Pi_t)\ud \hat{B}_t +  \frac{\Delta\mu (1-\Pi_{t^-})\Pi_{t^-}}  {\mu_\ell + \Delta \mu \Pi_{t^-}}  \ud \hat m_t,\end{align}
with initial condition $\Pi_0=\pi$.
\end{theorem}
\begin{proof}
First, we consider $\pi\in[0,1)$. From \eqref{Eq:varphi}, for any $t\ge0$, one has
\begin{align} \label{pi_tit0}
\Pi_t = \frac{\varphi_t}{1 + \varphi_t}.    
\end{align}
For any $x\ge0$, for $f(x)= \frac{x}{1 + x}$ we have $f^{\prime}(x)= \frac{1}{(1 + x)^2}$ and $f^{\prime \prime}(x)= - \frac{2}{(1 + x)^3}$. Applying It\^o's formula to \eqref{pi_tit0} and using \eqref{dvarphieq}, it yields
\begin{equation} \label{filtersdeeq}
\ud \Pi_t = \frac{1}{(1 + \varphi_t)^2} \ud \varphi^c_t - \frac{1}{2} \frac{2}{(1 + \varphi_t)^3} \frac{\Delta \mu^2}{\beta^2} \varphi^2_t \ud t + \ud \bigg(\sum_{u\leq t}\Big [\frac{\varphi_u} {1 + \varphi_u} - \frac{\varphi_{u^-}}{1 + \varphi_{u^-}} \Big ] \bigg),
\end{equation}
where $(\varphi^c_t)_{t\ge0}$ denotes the continuous part of $(\varphi_t)_{t\ge0}$, given by
$$ \ud \varphi^c_t = \lambda( 1 + \varphi_t) \ud t + \frac{\Delta \mu}{\beta^2}  \varphi_t \ud Y_t - \varphi_{t^-} \Delta \mu (1-H_{t^-})  \ud t.$$
Using \eqref{innovationB} and the easily verifiable equalities
$1 + \varphi_t = \frac{1}{1-\Pi_t}$, $\frac{\varphi_t}{(1 + \varphi_t)^2} = \Pi_t (1 -\Pi_t)$ and $\frac{\varphi^2_t}{(1 + \varphi_t)^3} = \Pi^2_t (1- \Pi_t)$, we get
\begin{align}\label{C4}
        \begin{aligned}
        &\frac {1}{(1 + \varphi_t)^2} \ud \varphi^c_t - \frac{1}{2} \frac{2}{(1 + \varphi_t)^3} \frac{\Delta \mu^2}{\beta^2} \varphi^2_t \ud t \\
        &\quad  = \lambda(1-\Pi_t)\ud t+\frac{\Delta\mu}{\beta}\Pi_t(1-\Pi_t)\ud \hat{B}_t -  \Delta\mu (1-\Pi_t)\Pi_t (1-H_{t^-})\ud t.
        \end{aligned}
\end{align}

It remains to compute
\begin{align} \label{C5}
\begin{aligned}
&\sum_{u\leq t}\Big [\frac{\varphi_u} {1 + \varphi_u} - \frac{\varphi_{u^-}}{1 + \varphi_{u^-}} \Big ]  = \int_0^t \left ( \frac{\varphi_{u^-}( 1 + \frac{\Delta \mu}{\mu_\ell}) } {1 + \varphi_{u^-}( 1 + \frac{\Delta \mu}{\mu_\ell}) } - \Pi_{u^-} \right ) \ud H_u\\
& =  \int_0^t \left ( \frac{\Pi_{u^-}( 1 + \frac{\Delta \mu}{\mu_\ell}) }  {1 + \Pi_{u^-}\frac{\Delta \mu}{\mu_\ell} } - \Pi_{u^-} \right ) \ud H_u  = \int_0^t \left ( \frac{\mu_h\Pi_{u^-} }  {\mu_\ell + \Delta \mu\Pi_{u^-}} - \Pi_{u^-} \right ) \ud H_u.
 \end{aligned}
\end{align}
Plugging \eqref{C4} and \eqref{C5} into \eqref{filtersdeeq}, we obtain the filtering equation
$$\ud \Pi_t = \lambda(1-\Pi_t)\ud t+\frac{\Delta\mu}{\beta}\Pi_t(1-\Pi_t)\ud \hat{B}_t -  \Delta\mu (1-\Pi_t)\Pi_t (1-H_{t^-})\ud t + \Big ( \frac{\mu_h\Pi_{t^-} }  {\mu_\ell + \Delta \mu \Pi_{t^-}} - \Pi_{t^-} \Big ) \ud H_t.$$
We derive equation \eqref{SDEFilterG} observing that
$$\frac{\mu_h\Pi_{t^-} }  {\mu_\ell + \Delta \mu \Pi_{t^-}} - \Pi_{t^-}=
\frac{\Delta\mu (1-\Pi_{t^-})\Pi_{t^-}}{\mu_\ell + \Delta \mu \Pi_{t^-}}$$
and
$$\hat m_t=H_t  - \int_0^t (1 - H_{s^-})(\mu_\ell + \Delta \mu\Pi_{s^-}) \ud s.$$

Finally, we consider the case $\pi=1$. Since $P_\pi(\xi=0)=1$ we get $\Pi_t\equiv1$, which is a solution to \eqref{SDEFilterG}.
\end{proof}

\newpage

\begin{remark}\label{rem jump}

Equation \eqref{SDEFilterG} can be equivalently rewritten as:
\begin{itemize}
\item[(i)] for $t<\tau$
\begin{equation}\label{before}\ud \Pi_t=\lambda(1-\Pi_t)\ud t+\frac{\Delta\mu}{\beta}\Pi_t(1-\Pi_t)\ud \hat{B}_t  -  \Delta\mu (1-\Pi_t)\Pi_t \ud t, \quad \Pi_0=\pi\end{equation}

\item[(ii)] for $t=\tau$
\begin{equation} \label{KSjump} \Delta \Pi_{\tau} = \frac {\Delta\mu \Pi_{\tau^-} (1 - \Pi_{\tau^-}) }{\mu_\ell + \Delta\mu \Pi_{\tau^-} } = \frac {\mu_h \Pi_{\tau^-} }{\mu_\ell + \Delta\mu \Pi_{\tau^-} } - \Pi_{\tau^-}\end{equation}
    
\item[(iii)]  for $t > \tau$
\begin{equation} \label{after}\Pi_t = \Pi_{\tau} + \int_{\tau}^t \lambda(1-\Pi_s) \ud s +
\int_{\tau}^t \frac{\Delta\mu}{\beta}\Pi_s(1-\Pi_s)\ud \hat{B}_s.\end{equation}
\end{itemize}
\end{remark}

 \begin{proposition}
     The filter is the unique strong solution to the SDE \eqref{SDEFilterG}.
 \end{proposition} 

 \begin{proof}

From Remark \ref{rem jump}, observe that, before and after the jump time $\tau$, the filter solves two diffusion equations \eqref{before} and \eqref{after}, both having the same diffusion coefficient given by  
\[
\sigma(x)  \coloneqq  \frac{\Delta\mu}{\beta} x(1-x).
\]
Denote with $b_0(x)$ and $b_1(x)$ the drifts of the SDEs \eqref{before} and \eqref{after} given respectively by  
\[
b_0(x)  \coloneqq  (1-x)(\lambda - \Delta\mu \cdot x), \quad b_1(x) \coloneqq  \lambda (1-x).
\]
We observe that $b_0(x)$, $b_1(x)$, and $\sigma(x)$ are continuous functions on $\mathbb{R}$, and they satisfy a local Lipschitz continuity condition. Thus, by \cite[Theorem 3.1, p. 164]{ikeda1981stochastic} we get the uniqueness for \eqref{before} with initial condition $\pi \in [0,1]$. 
The filter jumps at $\tau$ and its value is given by
$$\Pi_{\tau} = \frac {\mu_h \Pi_{\tau^-} }{\mu_\ell + \Delta\mu \Pi_{\tau^-} }.$$
This value is completely determined by the value of the filter in the interval $[0,\tau)$ because
$\Pi_{\tau^-} = \lim_{t \to \tau^-} \Pi_t$. 
Now, again from \cite[Theorem 3.1, p. 164]{ikeda1981stochastic}, we get uniqueness for \eqref{after} that concludes the proof.
 \end{proof}

\subsection{Sensitivity Analysis and Comparison of Hazard Rate Estimation Approaches}

In this section, we analyse the sensitivity of the filter dynamics with respect to model parameters. Moreover, we compare two estimation approaches for the hazard rate process, highlighting the impact of using different information flows.

\subsubsection{Sensitivity Analysis of the Filter Dynamics}
We examine the sensitivity of the filter dynamics in Equation \eqref{SDEFilterG} to changes in key parameters.
The numerical simulations are performed with arbitrarily chosen parameters to enhance the visual representation of different behaviours. Specifically, we set $\lambda=0.06$, $\mu_\ell=0.02$, $\pi=0$, and $T=60$, while varying $\beta$ and $\mu_h$ in the following scenarios:
\begin{itemize}
    \item Case A: $\beta = 1$ and $\mu_h=0.12$ (Figure \ref{fig:sub1})
    \item Case B: $\beta = 2$ and $\mu_h=0.12$ (Figure \ref{fig:sub2})
    \item Case C: $\beta = 2$ and $\mu_h=0.22$ (Figure \ref{fig:sub3})
\end{itemize}

\begin{figure}[h]
    \centering
    \begin{subfigure}[b]{0.495\textwidth}
        \includegraphics[width=1\textwidth]{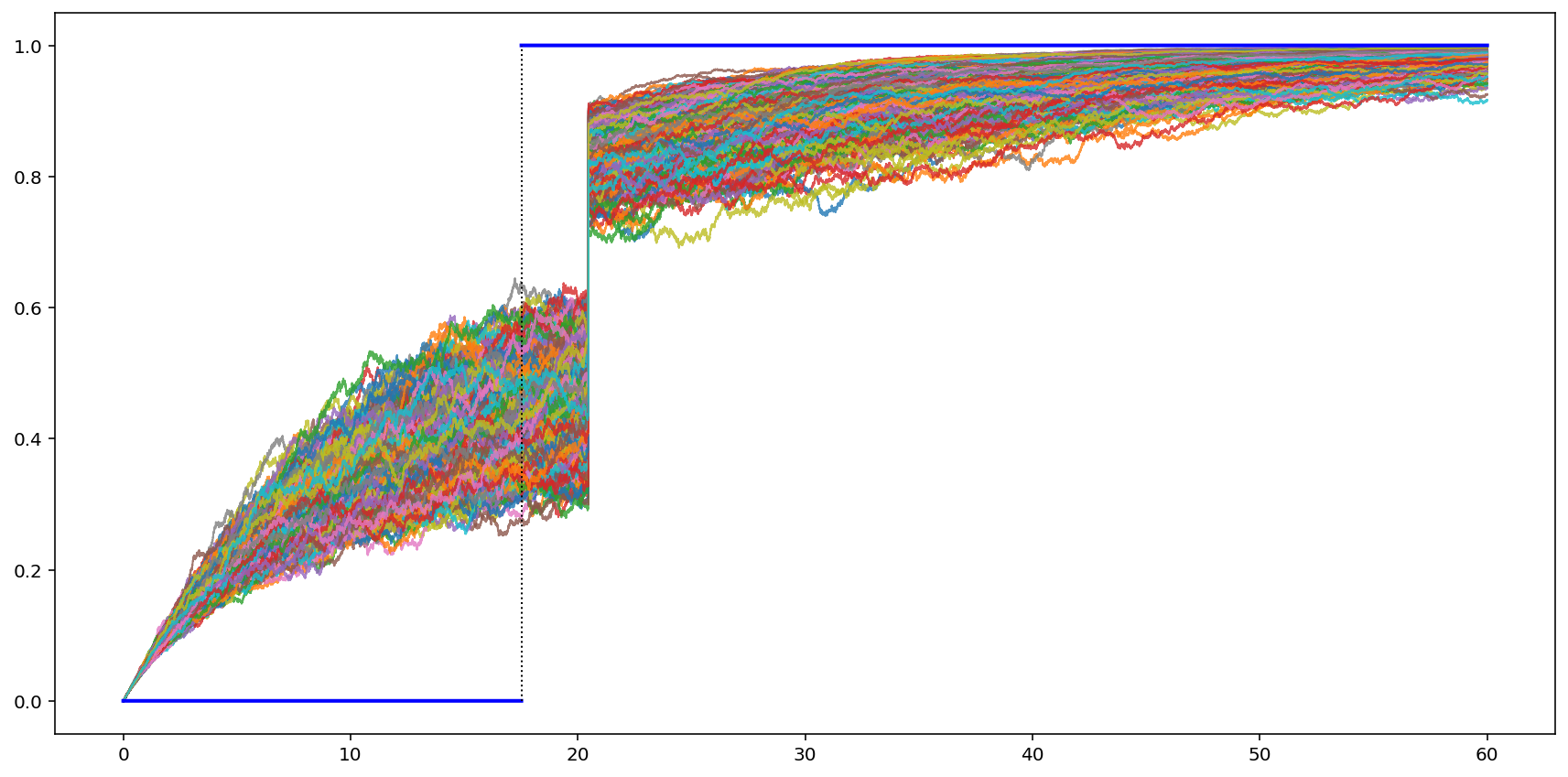}
        \caption{Case A: $\beta = 1$ and $\mu_h=0.12$}
        \label{fig:sub1}
    \end{subfigure}
    \hfill
    \begin{subfigure}[b]{0.495\textwidth}
        \includegraphics[width=1\textwidth]{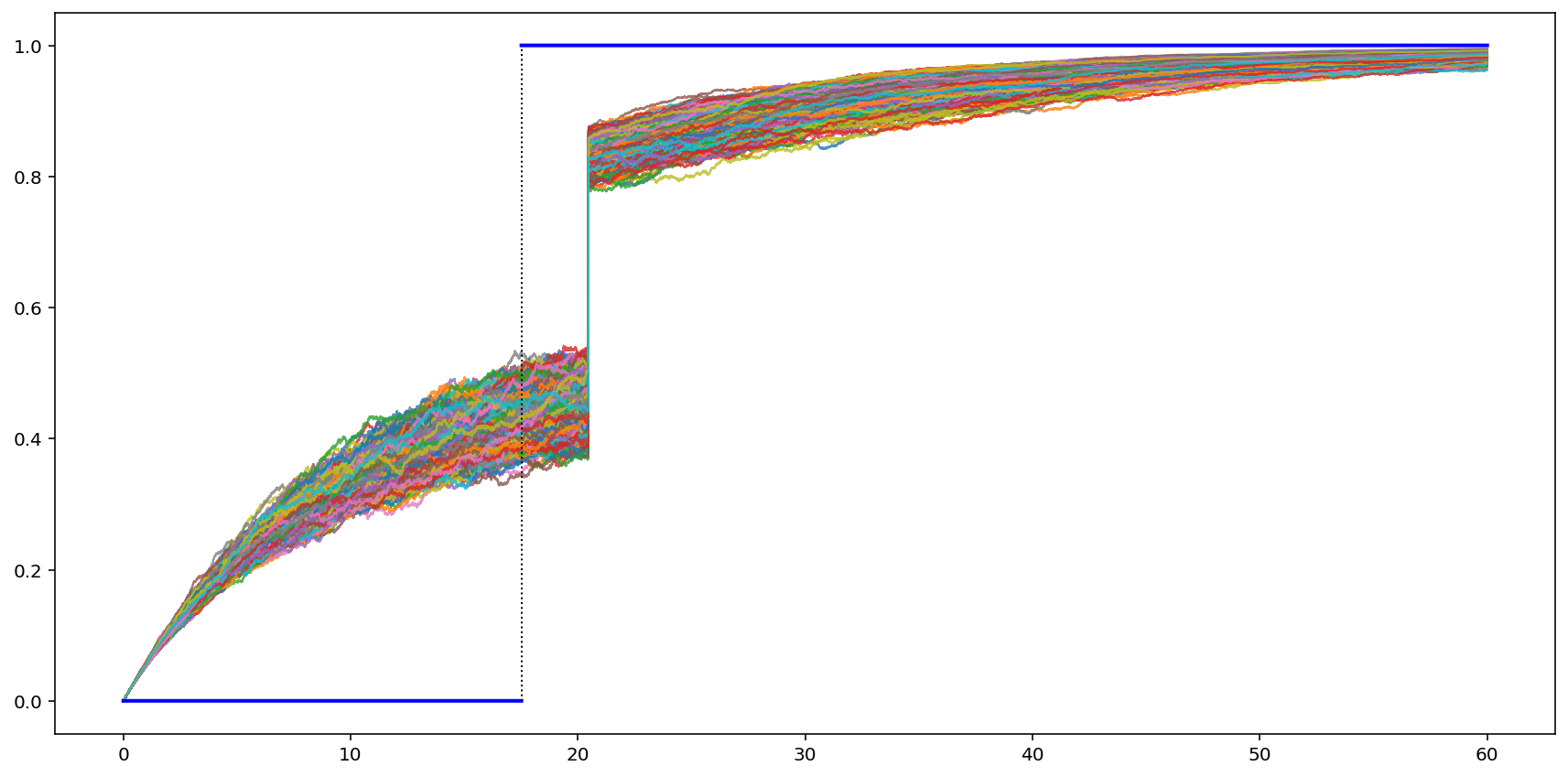}
        \caption{Case B: $\beta = 2$ and $\mu_h=0.12$}
        \label{fig:sub2}
    \end{subfigure}
    
    \vspace{0.5cm}
    
    \begin{subfigure}[b]{0.495\textwidth}
        \includegraphics[width=\textwidth]{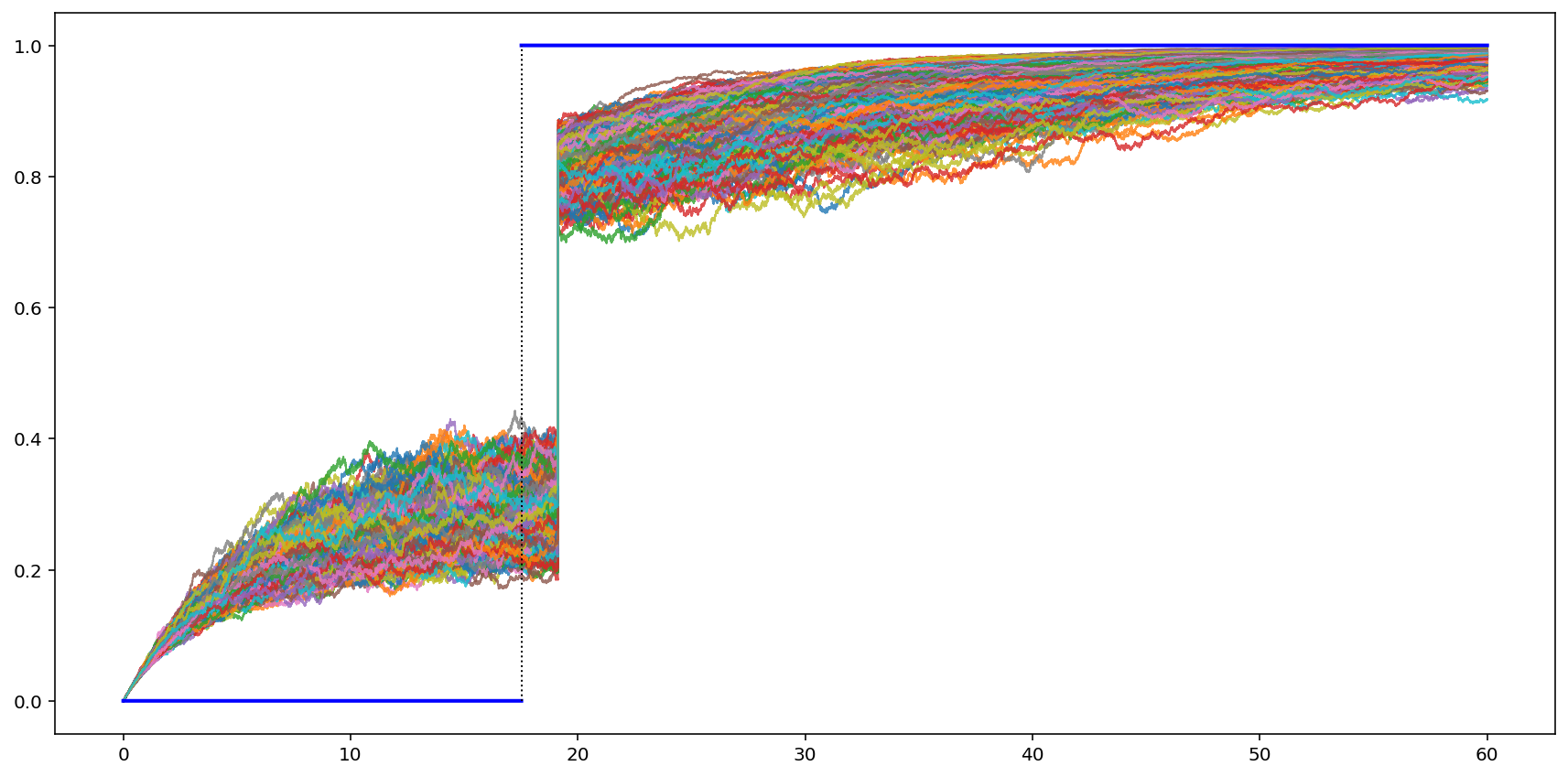}
        \caption{Case C: $\beta = 2$ and $\mu_h=0.22$}
        \label{fig:sub3}
    \end{subfigure}
    \caption{Trajectories of the filter dynamics for different parameter choices.}
\end{figure}

Figures \ref{fig:sub1}--\ref{fig:sub3} show the plot of $\mathbf{1}_{\{t<\xi\}}$ (in blue) and multiple sample paths (in different colors) of $\Pi_t$.
Due to the stochastic nature of the simulation, the values $\xi$, $\tau$ and trajectories of the process $\Pi_t$ vary each time the simulation is run. However, because our sensitivity analysis remains the same across different runs, we fix the seed to 18 to ensure the reproducibility of our results and interpretations. With this choice, we obtain $\xi=17.51$.
The value of $\tau$ depends on $\mu_h$. Specifically, for $\mu_h=0.12$ (cases A and B) we obtain $\tau=20.46$, whereas for $\mu_h=0.22$ (case C) the jump time occurs earlier at $\tau=19.12$. 

We first compare cases A and B to assess the impact of a higher $\beta$.  From figures \ref{fig:sub1}--\ref{fig:sub2}, observe that increasing $\beta$ attenuates Brownian fluctuations in filter dynamics. This is consistent with Equation \eqref{SDEFilterG}, where the diffusion term is scaled by $1/\beta$, reducing its impact as $\beta$ increases. 
However, a lower value of $\beta$ results in a more accurate estimate due to the additive Gaussian noise structure (cf. \eqref{Addnoise}).
To illustrate this point, we simulate 1000 trajectories of $\Pi_t$ and analyse the percentage of paths that remain close to $\mathbf{1}_{\{t<\xi\}}$ at specific time points. For $t<\xi$, greater accuracy corresponds to $\Pi_t$ staying near zero, whereas for $\xi\le t$, $\Pi_t$ is more accurate when it remains close to 1. 
At $t=\xi/2$, we find that the 45.4\% of trajectories for $\beta=1$ (case A) satisfy $\Pi_t < 0.3$, compared to only 35.2\% for $\beta=2$ (case B).
At $t=2\xi$, the 24.2\% of trajectories for $\beta=1$ exhibit $\Pi_t>0.95$ compared to 5.6\% for $\beta=2$.

Next, we analyse the effect of increasing $\mu_h$ by comparing cases B and C. As noted earlier, a higher $\mu_h$ results in an earlier jump time $\tau$. Moreover, it amplifies Brownian fluctuations, as $\mu_h$ directly influences the diffusion term in the SDE \eqref{SDEFilterG}. It also affects the drift term, impacting on the filter's accuracy. In case C, at $t=\xi/2$, 98.8\% of trajectories satisfy $\Pi_t<0.3$, while at $t=2\xi$, the 21.4\% of trajectories exhibit $\Pi_t>0.95$.

\subsubsection{Comparison between the $\bG^Y$-estimate and the $\bF^Y$-estimate of $(\mu_t)_{t\ge 0}$}

We now examine how different information sets affect the estimation of the hazard rate process $(\mu_t)_{t\ge 0}$. Specifically, we compare the $\bG^Y$-estimate of the hazard rate process $(\mu_t)_{t\ge 0}$ with the case where the information available is restricted to $\bF^Y$.
Recall that the $\bG^Y$-estimate, denoted by $(\hat{\mu}_t)_{t\ge 0}$, incorporates both the noisy observations of the hazard rate (contained in $\bF^Y$) and the default-related information (contained in $\bH$). In contrast, the $\bF^Y$-estimate is based solely on the noisy observation of the hazard rate.

For the case with information restricted to $\bF^Y$, we define the $\bF^Y$-estimate as
\begin{equation}
\hat{\mu}^{F}_t  \coloneqq  \media_\pi [\mu_t \given \mathcal{F}^Y_t], \quad t\ge 0.
\end{equation}
To express this estimate, we introduce the filter 
\begin{align}
\Pi^{F}_t  \coloneqq  \prob_\pi (\xi \leq t \given \mathcal{F}^Y_t), \quad t\ge 0,
\end{align} 
Under this setup, $(\Pi^{F}_t)_{t\ge 0}$ is governed by the stochastic differential equation 
\begin{align} \label{SDEFilterY}
\ud \Pi^{F}_t = \lambda(1-\Pi^{F}_t)\ud t + \frac{\Delta\mu}{\beta}\Pi^{F}_t(1-\Pi^{F}_t)\ud \hat B^{F}_t, \quad \Pi^{F}_0 = \pi,
\end{align} 
where $(\hat B^{F}_t)_{t\ge 0}$ is a standard $(\bF^Y,\prob_\pi)$-Brownian motion defined, for $t\ge 0$, by  
\begin{align}
\hat B^{F}_t  \coloneqq  \frac{1}{\beta}\bigg(Y_t - \Delta\mu \int_0^t \Pi^{F}_s \, \ud s \bigg).
\end{align} 
Further details are given in \cite[pp. 308-310]{peskir2006optimal}. The $\bF^Y$-estimate of $(\mu_t)_{t\ge 0}$ takes the form \begin{equation} \label{F-estimate} \hat \mu^{F}_t = \media_\pi[ \mu_t \given  \cF^Y_t] = \mu_\ell (1- \Pi^{F}_t) + \mu_h  \Pi^{F}_t =  \mu_\ell + \Delta \mu \Pi^{F}_t \quad t \geq 0.\end{equation}
By the properties of conditional expectation, we have that for any fixed $t\in [0,T]$ the mean square error of the $\bG^Y$-estimate is less or equal than the mean square error of the $\bF^Y$-estimate, indeed  
$$\media_\pi[ (\hat \mu_t - \mu_t)^2] = \min _{Z_t^G \in L^{2}(\cG^Y_t)}\media_\pi[(Z_t^G-\mu_t)^2]] \leq \min _{Z_t^F \in L^{2}(\cF^Y_t)}\media_\pi[(Z_t^F-\mu_t)^2]]= \media_\pi[ (\hat \mu^{F}_t - \mu_t)^2],$$
where $L^{2}(\cG^Y_t)$ $\big(L^{2}(\cF^Y_t)\big)$ is the set of square-integrable random variables $\cG^Y_t$-measurable ($\cF^Y_t$-measurable) and clearly $L^{2}(\cF^Y_t)\subseteq L^{2}(\cG^Y_t)$.

\begin{table}[ht]
\centering
\begin{tabular}{|c|}
\hline
Parameters \\
\hline
$\pi=0$, $T=60$, $\lambda=0.06$, $\mu_\ell=0.02$, $\mu_h=0.22$, $\beta=1$\\
\hline
\end{tabular}
\caption{Parameters used to plot Figure \ref{fig:GFEstimate}.}
\label{tab:parameters}
\end{table}

To illustrate the pathwise behaviour of the two estimates, we perform a numerical analysis with the parameters summarized in Table \ref{tab:parameters}.
Notice that with $\pi=0$, we have $\prob_\pi(\xi>0)=1$. 
In Figure \ref{fig:GFEstimate}, the unobservable hazard rate $(\mu_t)_{t\ge 0}$ is shown in blue. 
It remains constant at $\mu_\ell$ until the jump at time $\xi>0$, after which it increases to $\mu_h$. 
The yellow trajectories represent the $\bG^Y$-estimate, which incorporates both the noisy observation of the hazard rate process and the default-related information. These trajectories exhibit a sudden change at the default time $\tau$, reflecting the immediate update triggered by the default event. In contrast, the green trajectories represent the $\bF^Y$-estimate, which is based solely on the noisy observation and evolves smoothly without a sudden jump.  
By comparing the trajectories of $(\hat{\mu}^F_t)_{t\ge 0}$ (the $\bF^Y$-estimate) and $(\hat{\mu}_t)_{t\ge 0}$ (the $\bG^Y$-estimate), we observe that before the time $\xi$, the $\bG^Y$-estimate is generally more accurate, as its trajectory remains closer to $\mu_\ell$. However, in the period between $\xi$ and $\tau$, the $\bF^Y$-estimate provides a more reliable estimate. 
After default, the $\bG^Y$-estimate quickly realigns with $(\mu_t)_{t\ge 0}$, demonstrating the advantage of incorporating full default-related data.

\begin{figure}
    \centering
    \includegraphics[width=0.75\linewidth]{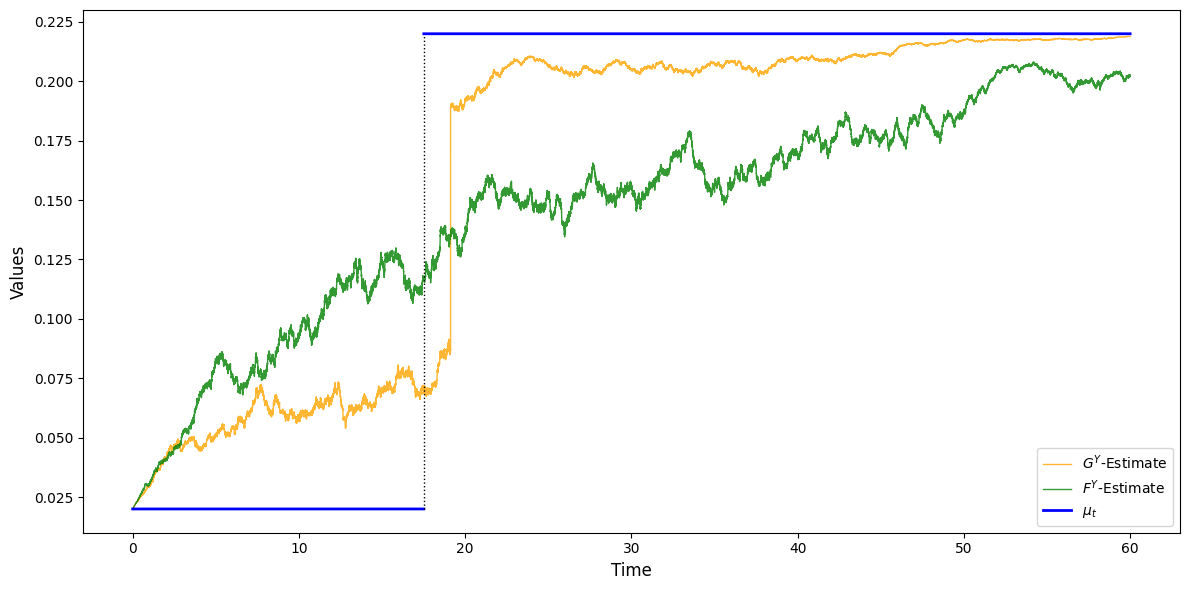}
    \caption{Plot of the hazard rate $(\mu_t)_{t\ge 0}$ (in blue), its $\bG^Y$-estimate $(\hat \mu_t)_{t\ge0}$ (in yellow) and its $\bF^Y$-estimate $(\hat \mu^F_t)_{t\ge0}$ (in green).}
    \label{fig:GFEstimate}
\end{figure}

\section{Financial and actuarial applications} \label{Sec:finapp}
In this section, we illustrate the practical relevance of our filtering framework by applying it to the pricing of credit-sensitive instruments and life-insurance contracts under partial information.
Precisely, in the next section, we derive a closed-form expression for the conditional survival probability under restricted information (see Proposition \ref{Prop:conditional}), which serves as a key building block for pricing credit derivatives and life-insurance contracts. Building on this result, in Section \ref{Sec:Creditd}, we obtain explicit pricing formulas for instruments such as defaultable bonds and credit default swaps. Finally, in Section \ref{Sec:extension}, we discuss an extension to price instruments not only by default/mortality events but also by other market factors.

From now on, we consider a finite time horizon $T$ and for notational simplicity, we adopt the convention that blackboard bold capital letters denote filtrations restricted to the finite time horizon $[0,T]$. For instance, we will write $\bG$ for $(\cG_t)_{t\in[0,T]}$.

\subsection{The conditional survival probability} \label{Sec:condprob}

From now on, we assume that the filtration $\bF$ takes the form 
\begin{equation} \label{choiceF} \bF = \bF^\mu \vee \bF^B, \end{equation} where $\bF^\mu = (\mathcal F_t^\mu)_{t \ge 0}$ denotes the natural filtration generated by the hazard rate process $(\mu_t)_{t \ge 0}$, and $\bF^B = (\mathcal F_t^B)_{t \ge 0}$ is the filtration generated by the Brownian motion $(B_t)_{t \ge 0}$ that drives the observation process in Eq. \eqref{Addnoise}. Notice that since $\xi$ and $(B_t)_{t \ge 0}$ are independent, the filtration $\bF$ is right-continuous, see \cite[Ch.\ 1.1.4, Prop.\ 1.12, p.\ 6]{aksamit2017enlargement}.

We next present two complementary approaches for computing the conditional survival probability under partial information, i.e. $\prob_\pi(\tau > T \mid \mathcal G_t^Y)$ for $t\le T$.
The first approach exploits the $(\bG, \prob_\pi)$-Markovian structure of the pair $(\mu_t, H_t)_{t\in[0,T]}$, while the second approach leverages the  $(\bG^Y, \prob_\pi)$-Markovian structure of $(\Pi_t, H_t)_{t\in[0,T]}$. 
An advantage of the second approach is that it remains valid even without assuming condition  \eqref{choiceF}.
On the other hand, the first approach also allows us to compute the conditional survival probability under full information $\prob_\pi(\tau > T \given \cG_t)$ for $t\leq T$. Therefore, in order to compare the results between the full- and partial-information cases, we work under \eqref{choiceF}.

We first state a preliminary lemma, whose proof is postponed in Appendix \ref{App:proofs}.
\begin{lemma}\label{Lemma4.1}
   Let $H^\xi_t  \coloneqq  \ind_{\{\xi \leq t\}}$ for any $t \in [0,T]$. Then the process
   \begin{equation}
      M^\xi_t= H^\xi_t - \int_0^t \lambda (1-H^\xi_{s^-})\ud s, \quad t\in [0,T],
   \end{equation}
   is a $(\bG,\prob_\pi)$-martingale.
\end{lemma}

We now derive the Markov generator of the pair $(\mu_t,H_t)_{t\in[0,T]}$.

\begin{proposition}
  The pair $(\mu_t, H_t)_{t\in[0,T]}$ is a $(\bG, \prob_\pi )$-Markov process with generator 
  \begin{align} \begin{aligned}
  \mathcal L^{(\mu,H)} f(t, x, h) &= \frac{\partial f}{\partial t}(t,x,h) + \lambda[ f(t, x + \Delta \mu, h) - f(t, x, h)] \ind_{\{x=\mu_\ell\}} \\
  &\quad + x [ f(t, x , h + 1) - f(t, x, h)] \ind_{\{h=0\}},
  \end{aligned} \end{align} 
  for any measurable function $f(t, x, h)$, $(t,x,h) \in [0,T]\times \{\mu_\ell, \mu_h\} \times \{0,1\}$, $C^1$ on $t$.
\end{proposition}

\begin{proof}
First, notice that $\prob_\pi(\tau=\xi)=0$ from the {\it avoidance} of $\bF$-stopping time, see Remark \ref{Rm:avoidance}.
Let $f(t, x, h)$, $(t,x,h) \in [0,T]\times \{\mu_\ell, \mu_h\} \times \{0,1\}$, be a measurable function having a continuous derivative w.r.t. time.
The It\^o's formula gives
\begin{align} 
\begin{aligned} \label{Eq:itof}
f(t,\mu_t,H_t) &= f(0,\mu_0,H_0)+\int_0^t \frac{\partial f}{\partial t}(s,\mu_s,H_s) \ud s  + \int_0^t \big( f(s,\mu_{s^-}, H_s) - f(s,\mu_{s^-}, H_{s^-}) \big) \ud H_s\\
&\quad + \int_0^t \big(  f(s,\mu_s, H_{s^-}) - f(s,\mu_{s^-}, H_{s^-})\big) dH^\xi_s.
\end{aligned}
\end{align} 
Define, for $t\le T$,
\begin{align} 
\begin{aligned} \label{Mt-mg}
M_t  \coloneqq  & \int_0^t [f(s, \mu_{s^-} + \Delta\mu, H_{s-}) - f(s, \mu_{s^-}, H_{s-})(dH^\xi_s - \lambda(1-H^\xi_{s^-}) \ud s) \\
& + \int_0^t [f(s, \mu_{s^-}, H_{s-}+1) - f(s, \mu_{s^-}, H_{s^-})(\ud H_s - \mu_{s^-}(1-H_{s^-}) \ud s).
\end{aligned}\end{align} 
It is easy to see that
 \begin{align} 
\begin{aligned}
 &\media_\pi\Big[\int_0^T |f(s, \mu_{s^-} + \Delta\mu, H_{s-}) - f(s, \mu_{s^-}, H_{s-})|\lambda(1-H^\xi_{s^-}) \ud s \Big]< \infty\\
&\media_\pi\Big[\int_0^T |f(s, \mu_{s^-}, H_{s^-}+1) - f(s, \mu_{s^-}, H_{s^-})|\mu_{s^-}(1-H_{s-}) \ud s\Big] < \infty\end{aligned}\end{align}  
because $f(t,x,h)$ is bounded.
Thus, $(M_t)_{t\in[0,T]}$ is a $(\bG, \prob_\pi )$-martingale and Equation 
\eqref{Eq:itof} rewrites as
\begin{align} 
\begin{aligned} 
f(t,\mu_t,H_t) &= f(0,\mu_0,H_0)+\int_0^t \frac{\partial f}{\partial t}(s,\mu_s,H_s) \ud s + M_t\\
&\quad + \int_0^t \big( f(s,\mu_{s-}, H_{s-}+1) - f(s,\mu_{s-}, H_{s-}) \big) \mu_{s-}(1-H_{s-}) \ud s\\
&\quad + \int_0^t \big(  f(s,\mu_{s-}+\Delta\mu, H_{s-}) - f(t,\mu_{s-}, H_{s-})\big) \lambda(1-H^\xi_{s-}) \ud s. 
\end{aligned}
\end{align} 
Since
$(1-H^\xi_{s-})=\ind_{\{s\leq \xi\}}=\ind_{\{\mu_{s-}=\mu_\ell \}}$ and $(1-H_{s-})=\ind_{\{s\leq \tau\}}=\ind_{\{H_{s-}=0 \}}$, we obtain
\begin{equation}
    f(t,\mu_t,H_t) = f(0,\mu_0,H_0)+\int_0^t \mathcal L^{(\mu,H)} f(s,\mu_{s-},H_{s-}) \ud s + M_t.
\end{equation}
Finally, from \cite[Proposition 1.7, Chapter IV]{ethier2009markov}, the thesis follows.
\end{proof}

From the $(\bG,\prob_\pi)$-Markov property of $(\mu_t, H_t)_{t\ge 0}$, it follows that there exists a measurable function $f(t,x,h)$, $(t,x,h) \in [0,T]\times\{\mu_\ell, \mu_h\} \times \{0,1\}$ such that, for all $t\le T$,
\begin{align} \label{DEF_f}
    \prob_\pi(\tau > T \given \cG_t)=f(t,\mu_t, H_t), \quad \prob_\pi-a.s.
\end{align}
The function $f$ can be characterized as a solution of a PDE with a final condition. Precisely, we have the following result.
\begin{proposition}\label{P1}
Let $f(t,x,h)$, $(t,x,h) \in [0, T]\times \{\mu_\ell, \mu_h\} \times \{0,1\}$ be a measurable function having continuous derivative w.r.t. time, and solution to 
\begin{equation} \label{PDE1approach}
\mathcal L^{(\mu,H)} f(t, x,h)=  0, \quad  f(T, x, h) = 1-h.  
\end{equation} 
Then, for any $t \in [0,T]$ and $\prob_\pi$-a.s., 
\begin{equation} \label{PDE1approachComplete}
\prob_\pi(\tau > T \given \cG_t)=f(t,\mu_t, H_t),
\end{equation}
and
\begin{equation}\label{PDE1approachINComplete}
\prob_\pi(\tau > T \given \cG^Y_t)= f(t,\mu_\ell, H_t) (1-\Pi_t) +  f(t,\mu_h, H_t) \Pi_t.    
\end{equation}
\end{proposition}
\begin{proof}

Let $f(t,x,h)$, $(t,x,h) \in [0, T]\times \{\mu_\ell, \mu_h\} \times \{0,1\}$ be a solution to \eqref{PDE1approach}, by It\^o's formula, for any $t\le T$ 
\begin{align} 
\begin{aligned}\label{ito1}
  1-H_T=f(T,\mu_T,H_T) & =  f(t,\mu_t,H_t) + \int_t^T \mathcal L^{(\mu,H)} f(s, \mu_s,H_s)  \ud s + M_T - M_t\\
  & = f(t,\mu_t,H_t) + M_T - M_t,
 \end{aligned}\end{align}
where the $(\bG, \prob_\pi )$-martingale $(M_t)_{t\in[0,T]}$ is defined in \eqref{Mt-mg}.
Taking the conditional expectation to $\cG_t $ in \eqref{ito1}, we obtain \eqref{PDE1approachComplete}.

To obtain \eqref{PDE1approachINComplete}, we recall that
$\Pi_t=\prob_\pi(\mu_t=\mu_h\given \cG^Y_t)$ and $1-\Pi_t=\prob_\pi(\mu_t=\mu_\ell\given \cG^Y_t)$ for $t\le T$. By tower property and using \eqref{PDE1approachComplete}, for any $t \in [0,T]$
\begin{align}\begin{aligned} \label{EQ1a_fin}
    \prob_\pi(\tau > T \given \cG^Y_t)&=\media_\pi\Big[f(t,\mu_t, H_t) \given \cG^Y_t \Big]\\
    &= f(t,\mu_\ell, H_t) \prob_\pi(\mu_t=\mu_\ell\given \cG^Y_t) +  f(t,\mu_h, H_t) \prob_\pi(\mu_t=\mu_h\given \cG^Y_t) \\
    & = f(t,\mu_\ell, H_t) (1-\Pi_t) +  f(t,\mu_h, H_t) \Pi_t.
\end{aligned}
\end{align}
The proof is complete.
\end{proof}

From \eqref{PDE1approachINComplete}, we observe that the conditional survival probability under partial information can be expressed as a measurable function of $(t,\Pi_t, H_t)$. This representation follows from the $(\bG^Y, \prob_\pi)$-Markovianity of the process $(\Pi_t, H_t)_{t\in[0,T]}$. In the next proposition, we use the filtering equation \eqref{SDEFilterG} to derive the generator of $(\Pi_t, H_t)_{t\in[0,T]}$. The proof is given in Appendix \ref{App:proofs}.

\begin{proposition} \label{MarkovPiH}
  The pair $(\Pi_t, H_t)_{t\in[0,T]}$ is a $(\bG^Y, \prob_\pi )$-Markov process with generator 
  \begin{align} 
\begin{aligned}
    &\mathcal{L}^{(\Pi,H)} g(t, x, h) \\
    &=  \frac{\partial g}{\partial t}(t,x,h) + (1-x)(\lambda -  \Delta\mu \ x (1-h))\frac{\partial g}{\partial x}(t,x,h)+ \frac{1}{2} \bigg (\frac{\Delta\mu}{\beta}\bigg)^2  x^2(1-x)^2 \frac{\partial^2 g}{\partial x^2}(t,x,h) \\
    & \quad + \bigg [ g \bigg(t, \frac{\mu_h x}{\mu_\ell + \Delta \mu x}, h+1 \bigg) - g(t, x, h) \bigg](1-h) (\mu_\ell+\Delta \mu x),
\end{aligned} 
\end{align}
for any function $g(t, x, h)$, $(t,x,h) \in [0,T]\times[0,1] \times \{0,1\}$, $C^1$ on $t\in[0,T]$ and $C^2$ on $x\in [0,1]$.
\end{proposition}

From the $(\bG^Y, \prob_\pi)$-Markov property of $(\Pi_t, H_t)_{t\in[0,T]}$, it follows that there exists a measurable function $g(t,x,h)$, $(t,x,h) \in [0, T]\times [0,1] \times \{0,1\}$ such that, for all $t\le T$,
\begin{align}
    \prob_\pi(\tau > T \given \cG^Y_t)=g(t,\Pi_t, H_t), \quad  \prob_\pi - a.s.
\end{align}
The function $g(t,x,h)$ can be characterized via a PDE with final condition, as stated in the following proposition. We omit the proof for brevity, as it follows by proceeding as in the Proof of Proposition \ref{P1}.
\begin{proposition}\label{P2}
Let $g(t,x,h)$, $(t,x,h) \in [0,T]\times [0,1] \times \{0,1\}$ be a measurable function $C^1$ on $t\in[0,T]$ and $C^2$ on $x\in [0,1]$, and solution to 
\begin{equation}\label{PDE filtro}
\mathcal L^{(\Pi,H)} g(t, x,h)=  0, \quad  g(T, x, h) = 1-h.  
\end{equation}
Then, for any $t \in [0,T]$ and $\prob_\pi$-a.s., $\prob_\pi(\tau > T \given \cG^Y_t)=g(t,\Pi_t, H_t)$.
\end{proposition}

In what follows, we apply Proposition \ref{P1} to derive the conditional survival probability under both partial and full information. An alternative derivation of the conditional survival probability under partial information can be obtained using Proposition \ref{P2}, as emphasized in Remark \ref{R2}.
\begin{proposition} \label{Prop:conditional} 
The survival conditional probabilities under partial and full information are given by:
\begin{itemize}
    \item[i)] if $\mu_h=\mu_\ell+\lambda$, then for $t\le T$,
    \begin{align}
    \begin{aligned}
     &\prob_\pi(\tau > T \given \cG^Y_t)=\ind_{\{\tau>t\}} \big(1+\lambda (T-t) (1- \Pi_t) \big)\e^{-\mu_h(T-t)},\\
     &\prob_\pi(\tau > T \given \cG_t)=\ind_{\{\tau>t\}} \big(1+\lambda (T-t) \ind_{\{\xi>t\}} \big)\e^{-\mu_h(T-t)},
    \end{aligned}
    \end{align}
    \item[ii)] if $\mu_h\neq \mu_\ell+\lambda$, then for $t\le T$,
    \begin{align}\begin{aligned}
    &\prob_\pi(\tau > T \given \cG^Y_t)=\ind_{\{\tau>t\}}\Big( \kappa (1-\Pi_t) \e^{-(\mu_\ell+\lambda)(T-t)}+ \big(1-\kappa(1-\Pi_t)\big) \e^{-\mu_h(T-t)}\Big),\\
    &\prob_\pi(\tau > T \given \cG_t)=\ind_{\{\tau>t\}}\Big( \kappa \ind_{\{\xi>t\}} \e^{-(\mu_\ell+\lambda)(T-t)}+ \big(1-\kappa\ind_{\{\xi>t\}}\big) \e^{-\mu_h(T-t)}\Big),
    \end{aligned}
    \end{align}
\end{itemize}
        with 
\begin{align} \label{kappa}
    \kappa \coloneqq \frac{\Delta\mu}{\Delta\mu-\lambda}.
\end{align}
\end{proposition}
\begin{proof}
In view of Proposition \ref{P1}, we show that the PDE \eqref{PDE1approach} admits a solution, which we can compute explicitly. Equation \eqref{PDE1approach} reads  
\begin{align} \label{PDE1approach2}
\begin{cases}
\begin{aligned}
\frac{\partial f}{\partial t}(t,x,h) &+ \lambda [ f(t, x + \Delta \mu, h) - f(t, x, h)] \ind_{\{x=\mu_\ell\}} \\ 
&+ x [ f(t, x , h + 1) - f(t, x, h)] \ind_{\{h=0\}} = 0,  
\end{aligned}
\\
f(T,x,h) = 1-h,    
\end{cases}
\end{align}  
for $x \in \{\mu_\ell, \mu_h\}$ and $h \in \{0,1\}$. The system \eqref{PDE1approach2} results in solving four nested ordinary differential equations (ODEs).  
We begin with the case $x = \mu_h$ and $h = 1$, for which we obtain:  
\begin{align}
\frac{\partial f}{\partial t}(t, \mu_h, 1) = 0, \ t \in (0,T) \quad f(T, \mu_h, 1) = 0.  
\end{align}  
Solving this equation yields  
\begin{align} \label{gfmu11}
f(t, \mu_h, 1) = 0, \quad \forall \ 0 \leq t \leq T.  
\end{align}  
Next, consider the case $x = \mu_\ell$ and $h = 1$. Using equation \eqref{gfmu11}, we get  
\begin{align}
\frac{\partial f}{\partial t}(t, \mu_\ell, 1) - \lambda f(t, \mu_\ell, 1) = 0, \quad f(T, \mu_\ell, 1) = 0. 
\end{align}  
Thus,
\begin{align} \label{gfmu01}
f(t, \mu_\ell, 1) = 0, \quad \forall \ 0 \leq  t \leq T.  
\end{align}  
Consider the case $x = \mu_h$ and $h = 0$. Using \eqref{gfmu11}, it yields
\begin{align} \label{gfmu1}
\frac{\partial f}{\partial t}(t,\mu_h,0) - \mu_h  f(t, \mu_h, 0) = 0,  \quad
f(T,\mu_h,0) = 1.    
\end{align} 
Its solution is
\begin{align} \label{gfmu10}
f(t, \mu_h, 0) = \e^{-\mu_h(T-t)}, \quad \forall \ t \leq T.  
\end{align}  
Finally, for $x = \mu_\ell$ and $h = 0$, we have
\begin{align} \label{gODE_mu0h0}
\begin{cases}
\begin{aligned}
\frac{\partial f}{\partial t}(t,\mu_\ell,0) &+ \lambda [ f(t, \mu_h , 0) - f(t, \mu_\ell, 0)] \\ 
&+ \mu_\ell [ f(t, \mu_\ell , 1) - f(t, \mu_\ell, 0)] = 0,   
\end{aligned}
\\
f(T,\mu_\ell,0) = 1.
\end{cases}
\end{align}  
Using \eqref{gfmu01} and \eqref{gfmu10}, equation \eqref{gODE_mu0h0} becomes
\begin{align}
\frac{\partial f}{\partial t}(t,\mu_\ell,0) &- (\lambda+\mu_\ell)f(t, \mu_\ell, 0)+ \lambda \e^{-\mu_h(T-t)}=0, \quad f(T,\mu_\ell,0) = 1.
\end{align}

\vspace{5mm}

Using the method of integrating factors for linear ODEs (see e.g. \cite{boyce2021elementary}), we find that 
\begin{itemize}
    \item[i)] if $\mu_h=\mu_\ell+\lambda$, then for $t\le T$,  
    \begin{align}
    f(t,\mu_\ell,0)=\big(1+\lambda(T-t)\big)\e^{-\mu_h(T-t)}.
    \end{align}
    \item[ii)] if $\mu_h\neq \mu_\ell+\lambda$, then for $t\le T$, 
    \begin{align}
    \begin{aligned}
    f(t,\mu_\ell,0)&= (1-\kappa) \e^{-\mu_h(T-t)} + \kappa \e^{-(\mu_\ell+\lambda)(T-t)}.
    \end{aligned}
    \end{align}
\end{itemize}

Finally, combining the above expressions, we get
\begin{itemize}
    \item[i)] if $\mu_h=\mu_\ell+\lambda$, then 
    \begin{align}\label{Eq:Probcompl1}
    f(t,x,h)=\ind_{\{h=0\}} \big(1+\ind_{\{x=\mu_\ell\}} \lambda (T-t)\big) \e^{-\mu_h(T-t)},
    \end{align}
    \item[ii)] if $\mu_h\neq \mu_\ell+\lambda$, then 
    \begin{align}\begin{aligned} \label{Eq:Probcompl2}
    f(t,x,h)= \ind_{\{h=0\}} \Big( (1-\ind_{\{x=\mu_\ell\}}\kappa) \e^{-\mu_h(T-t)} + \ind_{\{x=\mu_\ell\}} \kappa \e^{-(\mu_\ell+\lambda)(T-t)}\Big).
    \end{aligned}
    \end{align}
\end{itemize}
Finally, the proof is complete by applying \eqref{PDE1approachComplete} and \eqref{PDE1approachINComplete} from Proposition \ref{P1}.
\end{proof}

\begin{remark}\label{R2} An alternative derivation of the survival conditional probability under partial information is possible using Proposition \ref{P2}. Indeed, it is easy to see that the PDE \begin{equation}\label{PDE2 filtro}
\mathcal L^{(\Pi,H)} g(t, x,h)=  0, \quad  g(T, x, h) = 1-h,  
\end{equation}
with the ansatz $g(t, x, h) = f_0(t) (1-x) + f_1(t)x$ reduces to the nested ODEs in the Proof of Proposition \ref{Prop:conditional} leading to $f_0(t) = f(t,\mu_\ell,0)$ and $f_1(t)= f(t,\mu_h,0).$
\end{remark}

We conclude this section with the following remark, where we derive the conditional density of $\tau$ under partial and full information.
\begin{remark}
\label{densitytau}
Let $f_{\tau|\cG^Y_t}(s)$ and $f_{\tau|\cG_t}(s)$ for any $t\le s$ denote the conditional density of $\tau$ under partial and full information, respectively. From Proposition \ref{Prop:conditional}, it holds that:
\begin{itemize}
    \item[i)] If $\mu_h=\mu_\ell+\lambda$, then for $t\le s$
        \begin{align*}
        \begin{aligned}
         &f_{\tau|\cG^Y_t}(s)=\ind_{\{\tau>t\}} \e^{- \mu_h(s-t)} \bigg[\mu_h \big(1+\lambda (1-\Pi_t)(s-t)\big) - \lambda (1-\Pi_t) \bigg],\\
         &f_{\tau|\cG_t}(s)=\ind_{\{\tau>t\}} \e^{ -\mu_h(s-t)} \bigg[\mu_h\big(1+\lambda (s-t) \ind_{\{\xi>t\}}\big) -\lambda \ind_{\{\xi>t\}}  \bigg].
        \end{aligned}
    \end{align*}
    \item[ii)] If $\mu_h\neq \mu_\ell+\lambda$, then for $t\le s$
    \begin{align*}\begin{aligned}
    &f_{\tau|\cG^Y_t}(s)=\ind_{\{\tau>t\}}\bigg( (\mu_\ell+\lambda) \kappa (1-\Pi_t) \e^{-(\mu_\ell+\lambda)(s-t)}+\mu_h \big(1-\kappa(1-\Pi_t)\big) \e^{-\mu_h(s-t)}\bigg),\\
    &f_{\tau|\cG_t}(s)=\ind_{\{\tau>t\}}\bigg( (\mu_\ell+\lambda) \kappa \ind_{\{\xi>t\}} \e^{-(\mu_\ell+\lambda)(s-t)}+\mu_h \big(1-\kappa\ind_{\{\xi>t\}}\big) \e^{-\mu_h(s-t)}\bigg).
    \end{aligned}
    \end{align*}
\end{itemize}
\end{remark}

\subsection{Pricing of credit derivatives} \label{Sec:Creditd}
In this section, we develop a flexible framework for pricing a wide range of credit-sensitive instruments under partial information on the firm's hazard rate process. 
For a given $\pi\in[0,1]$, we assume that $\prob_\pi$ is the risk-neutral probability measure. To specify the cash flows associated with a defaultable claim, we introduce:
\begin{itemize}
    \item[-] an amount $L\in[0,\infty)$ paid at maturity $T$ if no default occurs before $T$,
    \item[-] a continuous premium (or coupon) payment $p:[0,\infty) \to [0,\infty)$ paid until default or maturity, whichever occurs first,
    \item[-] a recovery amount $W(\tau)$ paid at default, where $W:[0,\infty) \to [0,L]$ is a deterministic function, representing the fraction of $L$ recovered in the event of default.
\end{itemize}
Let $r:[0,\infty)\to(0,\infty)$ be the deterministic, time-dependent risk-free rate. 
The price of the credit-sensitive instrument at time $t \in [0,T]$ under partial information on the hazard rate of the firm is given by
\begin{align}
\begin{aligned} \label{pricecb}
P(t,T) :&= \media_\pi\Big[ \mathbf{1}_{\{\tau > T\}} \e^{-\int_t^T r(s) \ud s}L + \int_t^{\tau\wedge T} \e^{-\int_t^s r(u) \ud u} p(s) \ud s\\
&\quad + \mathbf{1}_{\{t< \tau \leq T\}} \e^{-\int_t^{\tau} r(s) \ud s} W(\tau) \given \mathcal{G}^Y_t \Big].
\end{aligned}
\end{align}

The first term of \eqref{pricecb} immediately reads 
\begin{align}
\media_\pi[ \mathbf{1}_{\{\tau > T\}} \e^{-\int_t^T r(s) \ud s} L \given \mathcal{G}^Y_t ]=\e^{-\int_t^T r(s) \ud s} L\  \prob_\pi(\tau > T \given \cG^Y_t),
\end{align}
where $\prob_\pi(\tau > T\given \cG^Y_t)$ is the conditional survival probability under partial information derived in Proposition \ref{Prop:conditional}.
The second term can be rewritten as follows
\begin{align}
\begin{aligned}
\media_\pi\Big[ \int_t^{\tau\wedge T} \e^{-\int_t^s r(u) \ud u} p(s) \ud s \given \mathcal{G}^Y_t \Big]&= \media_\pi\Big[ \int_t^{T} \ind_{\{\tau>s\}} \e^{-\int_t^s r(u) \ud u} p(s) \ud s \given \mathcal{G}^Y_t \Big] \\
&=\int_t^{T} \prob_\pi(\tau>s\given \cG^Y_t) \e^{-\int_t^s r(u) \ud u} p(s) \ud s.
\end{aligned}
\end{align}
Finally, the last term of \eqref{pricecb} becomes
\begin{align}
\begin{aligned}
\media_\pi\Big[ \mathbf{1}_{\{t< \tau \leq T\}} \e^{-\int_t^{\tau} r(u) \ud u}  W(\tau) \given \mathcal{G}^Y_t \Big] 
&=\int_t^T \e^{-\int_t^{s} r(u) \ud u} W(s) f_{\tau|\cG^Y_t}(s) \ud s,
\end{aligned}
\end{align}
where $f_{\tau|\cG^Y_t}(s)$ is the conditional density derived in Remark \ref{densitytau}.
Putting together the above considerations, the general pricing formula for a credit-sensitive instrument at time $t \in [0,T]$ under partial information is given by
\begin{align}
\begin{aligned} \label{pricingformulaCB}
P(t,T) &= \e^{-\int_t^T r(s) \ud s} L \prob_\pi(\tau > T \given \cG^Y_t)+ \int_t^{T} \prob_\pi(\tau>s\given \cG^Y_t) \e^{-\int_t^s r(u) \ud u} p(s) \ud s\\
&\quad + \int_t^T \e^{-\int_t^{s} r(u) \ud u} W(s) f_{\tau|\cG^Y_t}(s) \ud s.
\end{aligned}
\end{align}
Our framework can be tailored to price a wide range of credit-sensitive instruments by appropriately choosing the functions $p$ and $W$ and setting the parameter $L$ accordingly.
In the following, we illustrate its application to two common financial instruments.
\begin{itemize}
    \item \textbf{Defaultable coupon bond.} Consider a corporate bond that pays one monetary unit at maturity $T$ if no default occurs before $T$. However, should a default occur at a random time $\tau$ (with $\tau \le T$), the bondholder receives a recovery amount $W(\tau)$ paid at default. Furthermore, the bond pays a continuous coupon rate $p:[0,\infty)\to[0,1]$ until default or maturity, whichever comes first. Denoting $P_{CB}(t,T)$ the defaultable coupon bond's price and using the pricing formula \eqref{pricingformulaCB} with $L=1$, we obtain
\begin{align}
\begin{aligned} \label{Eq:pricingformulaCB}
P_{CB}(t,T) &= \e^{-\int_t^T r(s) \ud s}  \prob_\pi(\tau > T \given \cG^Y_t)+ \int_t^{T} \prob_\pi(\tau>s\given \cG^Y_t) \e^{-\int_t^s r(u) \ud u} p(s) \ud s\\
&\quad + \int_t^T \e^{-\int_t^{s} r(u) \ud u} W(s) f_{\tau|\cG^Y_t}(s) \ud s.
\end{aligned}
\end{align}

\item \textbf{Credit default swap.} Consider a credit default swap (CDS) where the protection buyer pays premium payments at rate $p(t)$ to the protection seller until default or maturity and receives a payment at default from the seller, $W(\tau)$, if the reference entity defaults before maturity.
Denoting $P_{CDS}(t,T)$ the price of the CDS and using \eqref{pricingformulaCB} with $L=0$,
\begin{align}
\begin{aligned} \label{cdsinsurance}
P_{CDS}(t,T) &= - \int_t^{T} \prob_\pi(\tau>s\given \cG^Y_t) \e^{-\int_t^s r(u) \ud u} p(s) \ud s\\
&\quad +\int_t^T \e^{-\int_t^{s} r(u) \ud u} W(s) f_{\tau|\cG^Y_t}(s) \ud s
\end{aligned}
\end{align}
\end{itemize}
Similar pricing formulas apply to life insurance contracts. For instance, \eqref{cdsinsurance} is suitable to price a contract where a payment is made upon the death of the insured (if it occurs before maturity), while the policyholder pays a continuous premium rate until death or maturity, whichever comes first.

\subsubsection{Numerical analysis}
In this section, we provide a numerical illustration of the pricing framework developed under partial information. We focus on defaultable zero-coupon bonds and compare their price under full and partial information, highlighting the impact of information asymmetry on their valuation.

Consider a defaultable zero-coupon bond (DZCB) with unitary face value and constant recovery rate $W(t)\equiv \delta\in[0,1]$ paid at default. Assume a constant risk-free rate $r>0$.
Denote with $P_{ZCB}(t,T)$ the price at time $t\in [0,T]$ of the DZCB under partial information, 
\begin{align}
\begin{aligned}
P_{ZCB}(t,T) & \coloneqq  \media_\pi[\e^{-r(T\wedge \tau-t)} (\delta \ind_{\{t<\tau \le T\}} + \ind_{\{\tau > T\}}) \given \cG^Y_t].
\end{aligned}
\end{align}
Using to the pricing formula \eqref{Eq:pricingformulaCB}, we get
\begin{align}
P_{ZCB}(t,T) \coloneqq \e^{-r(T-t)}\prob_\pi(\tau > T\given \cG^Y_t)+\delta \int_t^T  \e^{-r(s -t)} f_{\tau|\cG^Y_t}(s) \ud s.
\end{align}
From the explicit expressions for the conditional survival probability and the conditional density, provided in Proposition \ref{Prop:conditional} and Remark \ref{densitytau}, with some algebraic steps, we obtain explicit formulas for $P_{ZCB}(t,T)$.
\begin{itemize}
    \item[i)] If $\mu_h=\mu_\ell+\lambda$, then for $t\le T$
    \begin{align}
    \begin{aligned}
    P_{ZCB}(t,T)
    & =\mathbf{1}_{\{t<\tau\}}\left[1+\lambda (T-t)(1-\Pi_t)\left(1+\frac{\delta\mu_h}{r+\mu_h}\right)\right]\e^{-(r+\mu_h)(T-t)}\\
    &\quad +\mathbf{1}_{\{t<\tau\}}\frac{\delta}{r+\mu_h}\left[\mu_h-\lambda(1-\Pi_t)\left(1+\frac{\mu_h}{r+\mu_h}\right)\right]\left(1-\e^{-(r+\mu_h)(T-t)}\right).
    \end{aligned}
    \end{align}
    \item[ii)] If $\mu_h\neq \mu_\ell+\lambda$, then for $t\le T$
    \begin{align}
    \begin{aligned}
    P_{ZCB}(t,T)
    & = \mathbf{1}_{\{t<\tau\}}\kappa(1-\Pi_t)\left[\e^{-(r+\mu_\ell+\lambda)(T-t)}
    +\delta \frac{\mu_\ell+\lambda}{r+\mu_\ell+\lambda}\Bigl(1-\e^{-(r+\mu_\ell+\lambda)(T-t)}\Bigr) \right]\\
    &\quad +\mathbf{1}_{\{t<\tau\}}\Bigl(1-\kappa(1-\Pi_t)\Bigr)\left[
    \e^{-(r+\mu_h)(T-t)}
    +\delta\,\frac{\mu_h}{r+\mu_h}\Bigl(1-\e^{-(r+\mu_h)(T-t)}\Bigr)
    \right],
    \end{aligned}
    \end{align}
    with $\kappa$ defined in \eqref{kappa}.
\end{itemize}
To illustrate the impact of partial information on the price of a DZCB, we compare the partial information price $P_{ZCB}(t,T)$ with the full-information price $F_{ZCB}(t,T)$ defined as
\begin{align}
\begin{aligned}
F_{ZCB}(t,T) & \coloneqq  \media_\pi[\e^{-r(T\wedge \tau-t)} (\delta \ind_{\{t<\tau \le T\}} + \ind_{\{\tau > T\}}) \given \cG_t], \quad t\in[0,T].
\end{aligned}
\end{align}
Similar arguments as in the partial information case lead to
\begin{itemize}
    \item[i)] If $\mu_h=\mu_\ell+\lambda$, then for $t\le T$
    \begin{align}
    \begin{aligned}
    F_{ZCB}(t,T)
    &=\mathbf{1}_{\{t<\tau\}}\left[1+\lambda (T-t)\mathbf{1}_{\{t<\xi\}}\left(1+\frac{\delta\mu_h}{r+\mu_h}\right)\right]\e^{-(r+\mu_h)(T-t)}\\
    & \quad+\mathbf{1}_{\{t<\tau\}}\frac{\delta}{r+\mu_h}\left[\mu_h-\lambda\mathbf{1}_{\{t<\xi\}}\left(1+\frac{\mu_h}{r+\mu_h}\right)\right]\left(1-\e^{-(r+\mu_h)(T-t)}\right).
    \end{aligned}
    \end{align}
    \item[ii)] If $\mu_h\neq \mu_\ell+\lambda$, then for $t\le T$
    \begin{align}
    \begin{aligned}
    F_{ZCB}(t,T)
    &= \mathbf{1}_{\{t<\tau\}}\mathbf{1}_{\{t<\xi\}}\kappa\left[\e^{-(r+\mu_\ell+\lambda)(T-t)}
    +\delta \frac{\mu_\ell+\lambda}{r+\mu_\ell+\lambda}\Bigl(1-\e^{-(r+\mu_\ell+\lambda)(T-t)}\Bigr) \right]\\
    & \quad+\mathbf{1}_{\{t<\tau\}}\Bigl(1-\kappa\mathbf{1}_{\{t<\xi\}}\Bigr)\left[
    \e^{-(r+\mu_h)(T-t)}
    +\delta\,\frac{\mu_h}{r+\mu_h}\Bigl(1-\e^{-(r+\mu_h)(T-t)}\Bigr)
    \right].
    \end{aligned}
    \end{align}
\end{itemize}
We estimate the risk-free rate $r$ using monthly 3-month T-Bills data from 1990 to 2020, yielding $r = 0.0263$. The initial hazard rate level, $\mu_\ell$, is derived from ICE BofA BB US High Yield Index (Option-Adjusted Spread) and is estimated at $0.0366$, while the post-shock hazard rate level, $\mu_h$, is obtained from the ICE BofA CCC \& Lower US High Yield Index (Option-Adjusted Spread) yielding $0.1148$. We set $\lambda$ at $0.25$, implying that the jump occurs on average after 4 years. To enhance data visualization, we set the volatility of the noisy observation to $\beta=0.15$ and we analyse both the case without recovery ($\delta=0$) and with partial recovery ($\delta=0.5$). Moreover, we choose $\pi = 0$ implying that $\prob_\pi(\xi>0)=1$. 
Table \ref{tab:summary} summarizes these parameters along with additional statistical analysis.

\begin{table}[h]
\centering
\begin{tabular}{cccc}
\hline
\textbf{Parameter} & \textbf{Estimated Value} & \textbf{Std. Dev.} & \textbf{95\% CI} \\
\hline
$r$ & 0.0263 & 0.0222 & (0.0240, 0.0285) \\
$\mu_\ell$ & 0.0366 & 0.0183 & (0.0361, 0.0370) \\
$\mu_h$ & 0.1148 & 0.0539 & (0.1135, 0.1161) \\
\hline
\hline
\multicolumn{4}{l}{Other parameters: $\lambda = 0.25$, $\pi = 0$, $\beta = 0.15$, $T=10$, $\delta=0$ or $\delta=0.5$.} \\
\hline
\end{tabular}
\caption{Summary of Data and Statistical Analysis}
\label{tab:summary}
\end{table}
\begin{figure}[h!]
\centering

\begin{subfigure}[b]{\textwidth}
    \centering
    \includegraphics[width=\textwidth]{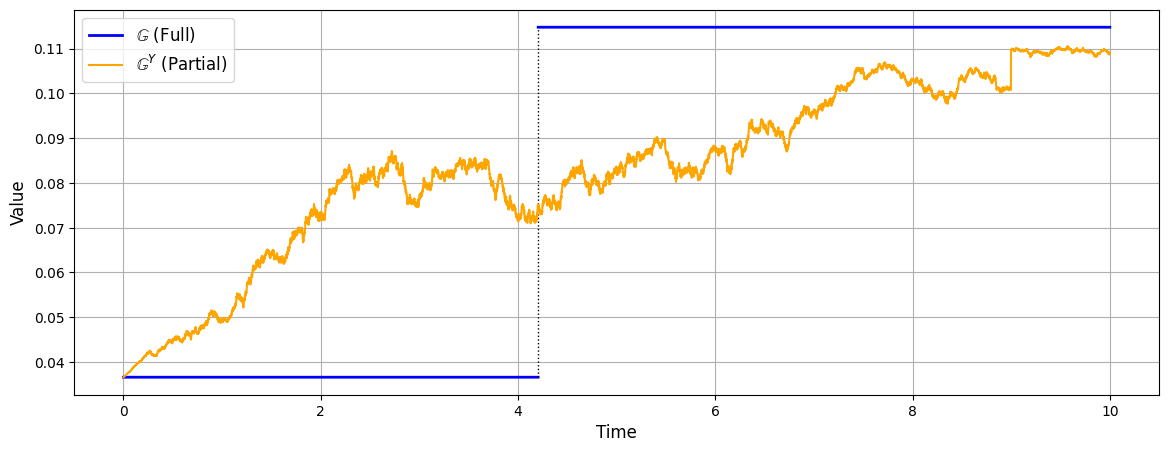}
    \caption{Evolution of the hazard rate process (blue) and its estimate under partial information (orange). The hazard rate jumps at $\xi$, while the estimated hazard rate adjusts gradually.}
    \label{fig:top}
\end{subfigure}

\begin{subfigure}[b]{\textwidth}
    \centering
    \includegraphics[width=\textwidth]{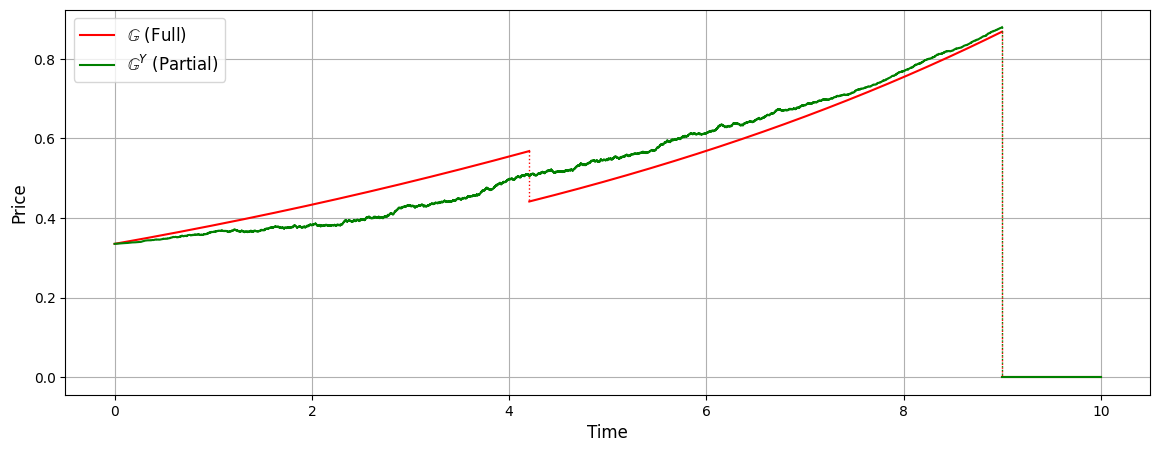}
    \caption{Price dynamics of a defaultable zero-coupon bond with no recovery under full information (red) and partial information (green).}
    \label{fig:middle}
\end{subfigure}

\begin{subfigure}[b]{\textwidth}
    \centering
    \includegraphics[width=\textwidth]{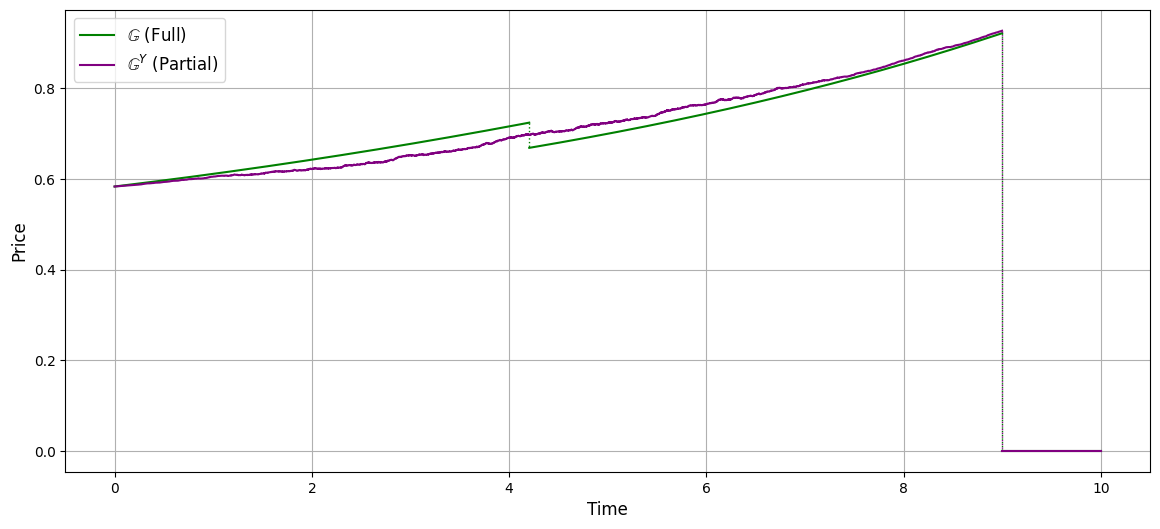}
    \caption{Price dynamics of a defaultable zero-coupon bond with a recovery rate of $\delta=0.5$ under full information (red) and partial information (green).}
    \label{fig:bottom}
\end{subfigure}
\caption{Comparison of the hazard rate process and defaultable zero-coupon bond prices under full and partial information.}
\label{fig:two_subfigures}
\end{figure}

Figure \ref{fig:top} shows the evolution of the hazard rate process $(\mu_t)_{t\in[0,T]}$ (blue curve) and its $\bG^Y$-estimate $(\hat \mu_t)_{t\in[0,T]}$ (orange curve) over time. As expected, the true hazard rate jumps at $\xi$ while the partial-information estimate adjusts more gradually and exhibits a jump at $\tau$. 
In figure \ref{fig:middle}, we plot the price of a defaultable zero-coupon bond with no recovery under full information (red curve) and partial information (green curve). 

The full-information price reacts immediately to changes in the hazard rate and thus shows jumps at both $\xi$ and $\tau$. These adjustments are marked by the vertical dotted lines in the figure.
By contrast, the partial-information price evolves continuously until default, since it does not instantaneously recognize the shock to the hazard rate and, analogously to the full-information price, jumps to zero at $\tau$. Before the shock in the hazard rate $(t<\xi)$, the partial information price is lower than the full information price due to an overestimation of default risk. Specifically, the partial-information estimate of the hazard rate $(\hat \mu_t)_{t\in[0,T]}$ remains above the true firm's hazard rate $(\mu_t)_{t\in[0,T]}$.
When the jump occurs at $\xi$, the price under full information immediately adjusts, and the price jumps down to reflect the increased risk of default. After the shock, the price under partial information remains above the full information price. 
In Figure \ref{fig:bottom}, we extend the analysis to a defaultable zero-coupon bond with recovery $\delta=0.5$. Here, the incorporation of recovery raises the bond price compared to the no-recovery case, as the bondholder is guaranteed to receive a fraction of the face value even in the event of default. The full-information price (green curve) continues to show sudden adjustments at $\xi$ and $\tau$, whereas the partial-information price (purple curve) evolves continuously until default.

\subsection{Extensions: an additional source of randomness}\label{Sec:extension}
We consider an extension of the baseline probability framework to include an additional source of randomness.
Many real-world contracts combine default/mortality risk with additional market factors. It is often reasonable to assume that these market factors are independent of either a counterparty's default or an individual's mortality. Examples may include:
\begin{itemize}
    \item Unit-linked life insurance contracts, which combine a pure endowment (paid at maturity if the insured is alive) and a term insurance benefit (paid at death if it occurs before maturity), with both components linked to the performance of an investment portfolio (see \cite{aase1994pricing}, \cite{moller1998risk}, \cite{CCC2} among others). 
    \item Vulnerable options whose value depends not only on the underlying asset's price dynamics but also on the creditworthiness of the counterparty (see \cite{johnson1987pricing}, \cite{jarrow1995pricing} among others).
\end{itemize}

To formalize this extension, we consider a probability space $(\widetilde\Omega,\widetilde{\cF}, \widetilde \prob)$ equipped with a filtration $\widetilde{\bF}$ satisfying the usual hypotheses of right continuity and $\widetilde \prob$-completeness. 
We then work on the product space
\begin{align}
    \big(\widetilde\Omega\times \Omega,\ \widetilde{\cF}\otimes\cF, \ \widetilde \prob\otimes \prob_\pi\big),
\end{align}
endowed with the right-continuous filtration $\widetilde{\bF}\vee \bF$. Within this enriched framework, we introduce a contingent claim whose payoff depends on whether $\tau$ occurs before the predetermined terminal time $T$, and on the evolution of an $\widetilde\bF$-adapted stochastic process $(X_t)_{t\in [0,T]}$. Let $\psi$ and $\phi$ be two measurable functions such that $\widetilde{\media}[\phi(X_T)] < \infty$  and for any $t \in [0,T]$, $\widetilde{\media}[\psi(X_t)]<\infty$.
If $\tau$ occurs before time $T$, the claim entitles the holder to a payment of $\psi(X_\tau)$, otherwise it pays $\phi(X_T)$ at $T$. 
Thus, the payoff of the claim is given by
\begin{align}
\mathbf{1}_{\{ \tau \le T\}}\,\psi(X_\tau) + \mathbf{1}_{\{\tau > T\}}\,\phi(X_T).  
\end{align} 
Our goal is to evaluate the expected discounted payoff at the current time, given the available information. In our setting, the investor/insurer has complete market-related information (captured by $\widetilde\bF$) but only partial default-related information (captured by $\bG^Y$). Accordingly, the price of the contingent claim at time $t\leq T$, under partial information on the default/mortality hazard rate, is defined as
\begin{align}
\begin{aligned}
J(t,T) \coloneqq \widetilde{\media}_\pi\bigg[&\mathbf{1}_{\{t< \tau \le T\}} \e^{-\int_t^\tau r(s) \ud s} \psi(X_\tau)+ \mathbf{1}_{\{\tau > T\}}\,\e^{-\int_t^T r(s) \ud s} \phi(X_T) \given \cG^Y_t\vee \widetilde \cF_t \bigg]
\end{aligned}
\end{align}
where $\widetilde{\media}_\pi$ denotes the expectation under the measure $\widetilde \prob\otimes \prob_\pi$ and $(r(t))_{t\in [0,T]} $ denotes the (deterministic) riskless interest rate.

By applying Fubini's Theorem, we derive the expression:
\begin{align}
\begin{aligned} \label{GenerReform}
    J(t,T)&= \int_t^T \e^{-\int_t^s r(u) \ud u} \widetilde{\media}[\psi(X_s) \given \widetilde \cF_t ] \, f_{\tau|\cG^Y_t}(s) \ud s + \e^{-\int_t^T r(u)\ud u} \widetilde{\media}\big[ \phi(X_T) \given \widetilde \cF_t \big] \prob_\pi (\tau>T\given \cG^Y_t).
\end{aligned}
\end{align}
This equation provides a pricing formula for a contingent claim under complete market-related information and partial default-related (or mortality-related) information. We note that, under the assumed independence, pricing in this framework reduces to pricing in a default-free market, adjusted by the conditional survival probability and density inferred from the partial default/mortality-related information.

\section{Conclusions}
In this paper, we develop a continuous-time framework for filtering in a hazard rate change-point model under partial information. Our approach combines noisy observations of the hazard rate with default-related information employing a progressive enlargement of filtration. 
By extending the change-of-measure techniques of \cite[pp. 308-310]{peskir2006optimal} to this setting, we derive a Kushner-Stratonovich type equation that estimates the conditional probability of the unobservable change-point.
Unlike the innovation approach, our methodology applies without requiring the hazard-rate process to be Markovian with respect to the underlying filtration.
We perform a sensitivity analysis of the filter dynamics with respect to key model parameters and provide a pathwise comparison under different information structures. In addition, we derive closed-form expressions for the conditional survival probability and the conditional density of the default time. These results are applied to the pricing of credit-sensitive instruments, including defaultable bonds and credit default swaps, in a partial information setting. We numerically compare defaultable zero-coupon bond prices under partial and full information, highlighting the effect of information asymmetry. The same formulas can be applied in the valuation of life-insurance contracts under restricted information on the mortality hazard rate of the insured.

Looking forward, future research could extend this framework to multi-event scenarios involving multiple change points or interacting hazard processes. Furthermore, incorporating additional sources of randomness that exhibit some dependence on default time may further enhance the model’s applicability in financial and insurance contexts.

\backmatter

\bmhead{Acknowledgements}
The second author was partially supported by INdAM-GNAMPA Project 2025 (CUP E5324001950001). Both authors were partially supported by the European Union – Next Generation EU – Project PRIN 2022 (code BEMMLZ), titled \textit{Stochastic Control and Games and the Role of Information}.

\begin{appendices}

\section{Preliminary proofs} \label{App:proofs}

\begin{proof}[\textbf{Proof of Proposition \ref{Prop:Probabilities}}]
The result is well known and established in \cite{shiryaev2007optimal} in a more general context. We also refer to \cite[Proposition 2.1]{deangelis2022quickestdetectionproblemfalse} for a proof similar to the one presented here. In our case, some additional care is required due to the specific structure of the filtration under consideration.

First, we show \eqref{Probabilities_and_Pi}. 
Take $A\in \cG^Y_t$, 
\begin{align}
    \begin{aligned}
        \media_\pi [\ind_A \Pi_t] &= \media_\pi [\ind_A \ind_{\{\xi\leq t\}}]\\
        &= \pi \media^0[\ind_A\ind_{\{\xi\leq t\}}]+ (1-\pi)\int_0^\infty  \media^s[\ind_A\ind_{\{\xi\leq t\}}] \lambda \e^{-\lambda s} \ud s\\
        &= \pi \media^0[\ind_A]+ (1-\pi)\int_0^t \media^s[\ind_A] \lambda \e^{-\lambda s} \ud s\\
        &= \media_\pi \bigg[\ind_A \bigg( \pi \frac{\ud \prob^0}{\ud \prob_\pi}\Big|_{ \cG^Y_t} + (1-\pi)\int_0^t \frac{\ud \prob^s}{\ud \prob_\pi}\Big|_{ \cG^Y_t} \lambda \e^{-\lambda s} \ud s \bigg) \bigg],
    \end{aligned}
\end{align}
where the first equality comes from the definition of $\Pi_t$, see \eqref{DefinitionPi}, the second equality use the definition of $\prob_\pi$ in \eqref{definitionProbPi}, the third one uses the definitions of $\prob^0$ and $\prob^s$ given in \eqref{definitionProbs}.

Next, we show \eqref{Probabilities_and_1-Pi}. 
First, we need to prove that for $s > t$\begin{align} \label{eqprob}
\prob^s\given \cG^Y_t=\prob^\infty \given \cG^Y_t.     
\end{align}
Observe that \begin{align}
\begin{aligned}
    \cG^Y_t=&\sigma((Y_{t_1}, H_{t_1})\in A_1,..,(Y_{t_n}, H_{t_n})\in A_n, \\
    &\quad t_1<...<t_n \in [0,t], \  A_1, ...A_n \in \mathcal{B}(\reali\times\{0,1\}), \ n\in \naturali),
\end{aligned}
\end{align}
and since $\reali\times\{0,1\}$ is separable, then $\mathcal{B}(\reali\times\{0,1\})=\mathcal{B}(\reali)\times\mathcal{B}(\{0,1\})$, see \cite[p. 244]{billingsley2013convergence}, it yields
\begin{align}
\begin{aligned}
    \cG^Y_t=&\sigma((Y_{t_1}, H_{t_1})\in C_1\times D_1,..,(Y_{t_n}, H_{t_n})\in C_n\times D_n, \\
    &\quad t_1<...<t_n \in [0,t], \  C_1, ...C_n \in \mathcal{B}(\reali), \ \  D_1, ...D_n \in \mathcal{B}(\{0,1\}),  \ n\in \naturali).
\end{aligned}
\end{align}
Now, take a set $A$ from the generator of $\cG^Y_t$,
\begin{align}
    \begin{aligned}
     \prob^s(A)= \prob^s((Y_{t_1}, H_{t_1})\in C_1\times D_1,..,(Y_{t_n}, H_{t_n})\in C_n\times D_n)
    \end{aligned}
\end{align}

Using the increasing property of $(\Lambda_t)_{t\ge 0}$, one gets $\{\tau>t\}=\{\Lambda_t<\Theta\}$ and, consequently $
H_t=\ind_{\{\Lambda_t\ge \Theta\}}$, $\prob_\pi$-a.s.
Moreover, from the definition of $\prob^s$ in \eqref{definitionProbs}, for $s > t>0$ it holds that
$H_t=\ind_{\{\mu_\ell t\ge \Theta\}}$ $\prob^s$-a.s.
Moreover, from \eqref{Eq:dY} we have that $Y_t=\beta B_t$ under $\prob^s$. Hence, denoting $Z_t \coloneqq \ind_{\{\mu_\ell t\ge \Theta\}}$,
\begin{align}
    \begin{aligned}
     \prob^s(A)&= \prob^s((\beta B_{t_1}, Z_{t_1})\in C_1\times D_1,..,(\beta B_{t_n}, Z_{t_n})\in C_n\times D_n)\\
     &=\mathsf{Q}((\beta B_{t_1}, Z_{t_1})\in C_1\times D_1,..,(\beta B_{t_n}, Z_{t_n})\in C_n\times D_n)
    \end{aligned}
\end{align}
where $\mathsf{Q}$ denotes the probability law of the process $(\beta B_{t}, Z_{t})_{t\ge0}$. By similar arguments,
\begin{align}
    \begin{aligned}
     \prob^\infty(A)=\mathsf{Q}((\beta B_{t_1}, Z_{t_1})\in C_1\times D_1,..,(\beta B_{t_n}, Z_{t_n})\in C_n\times D_n)
    \end{aligned}
\end{align}
from which it follows \eqref{eqprob}. To prove \eqref{Probabilities_and_1-Pi}, take $A\in  \cG^Y_t$,
\begin{align}
    \begin{aligned}
        \media_\pi [\ind_A (1-\Pi_t)] &= \media_\pi [\ind_A \ind_{\{\xi > t\}}]\\
        &=  (1-\pi)\int_0^\infty \media^s[\ind_A \ind_{\{\xi > t\}}] \lambda \e^{-\lambda s} \ud s\\
        &= \media_\pi \bigg[\ind_A  (1-\pi)\int_t^\infty \frac{\ud \prob^s}{\ud \prob_\pi}\Big|_{ \cG^Y_t} \lambda \e^{-\lambda s} \ud s  \bigg]\\
        &= \media_\pi \bigg[\ind_A  (1-\pi)\frac{\ud \prob^\infty}{\ud \prob_\pi}\Big|_{ \cG^Y_t}\int_t^\infty  \lambda \e^{-\lambda s} \ud s  \bigg]\\
        &= \media_\pi \bigg[\ind_A  (1-\pi)\e^{-\lambda t}\frac{\ud \prob^\infty}{\ud \prob_\pi}\Big|_{ \cG^Y_t}    \bigg],
    \end{aligned}
\end{align}
where in the fourth equality, we used \eqref{eqprob}. The proof is complete.

\end{proof}

\begin{proof}[\textbf{Proof of Proposition \ref{dvarphi}}]
Using \eqref{Probabilities_and_Pi} and \eqref{Probabilities_and_1-Pi} we get for any $t\ge0$

\begin{equation}\label{C0}
    \varphi_t = \frac{\pi}{1 - \pi} \e^{\lambda t} \frac{\ud \prob^0}{\ud \prob^\infty}\Big|_{ \cG^Y_t} + \e^{\lambda t} 
\int_0^t\frac{\ud \prob^s}{\ud \prob^\infty}\Big|_{ \cG^Y_t}\lambda \e^{-\lambda s}\ud s.\end{equation}
\\
We now focus on $Z_t \coloneqq  \frac{\ud \prob^0}{\ud \prob^\infty}\Big|_{ \cG^Y_t}$ and
$Z^s_t \coloneqq \frac{\ud \prob^s}{\ud \prob^\infty}\Big|_{ \cG^Y_t}$ for $s\le t$. Observe that, for $t$ fixed and $u\leq t$,
\begin{itemize}
    \item [(i)]
Under $\prob^0$, $\ud Y_u = \Delta \mu \ \ud u + \beta \ud B_u$ and $H_u$ has $(\cG^Y_u)_{u\leq t}$-predictable intensity $\mu_h(1-H_{u^-})$
\item [(ii)] Under $\prob^s$, $\ud Y_u = \Delta \mu \ind_{\{u\geq s \}} \ud u + \beta \ud B_u$ and $H_u$
has $( \cG^Y_u)_{u\leq t}$-predictable intensity $(\mu_\ell \ind_{\{u \leq s \}} + \mu_h \ind_{\{ u > s \}})( 1-H_{u^-})$
\item [(iii)] Under $\prob^\infty$, $\ud Y_u = \beta \ud B_u$ and $H_u$ has $( \cG^Y_u)_{u\leq t}$-predictable intensity $\mu_\ell(1-H_{u^-})$.
\end{itemize}

Thus by Girsanov's Theorem, 

\begin{equation} \label{zeta}
Z_t= \mathcal{E}\Big ( \int_0^t \frac{\Delta \mu }{\beta^2} \ud Y_u + \int_0^t  \frac{\Delta \mu }{\mu_\ell} \big (\ud H_u  -\mu_\ell(1-H_{u^-})  \ud u\big ) \Big ) = \mathcal{E}\big (M_t\big )
\end{equation}
where $\mathcal{E}$ denotes the Dol\'{e}ans-Dade exponential of the $( \bG^Y,\prob^\infty)$-martingale $(M_t)_{t\geq 0}$ given in \eqref{dM}, (see \cite[Theorem 4.61]{jacod2013limit}). In fact, observe that $(Z_u)_{u\ge 0}$ is a $( \bG^Y,\prob^\infty)$-martingale over any finite time horizon $t$ and under $\prob^0$ the $(\cG^Y_u)_{u\leq t}$-intensity of $H_u$ is given by 
$$\bigg(1 + \frac{\Delta \mu }{\mu_\ell}\bigg) \mu_\ell(1-H_{u^-}) = \mu_h  (1-H_{u^-})$$
and 
$$\frac{1}{\beta} (Y_u - \Delta \mu \ u)$$ 
is a $(\cG^Y_u)_{u\leq t}$-Brownian motion.
Similarly,   
\begin{equation}
Z^s_t= \mathcal{E}\Big ( \int_0^t \frac{\Delta \mu }{\beta^2}  \ind_{\{u \geq s \}}\ud Y_u + \int_0^t  \frac{\Delta \mu  }{\mu_\ell} \ind_{\{u \geq s \}}\big (\ud H_u  -\mu_\ell(1-H_{u^-})  \ud u\big ) \Big ).
\end{equation}
In fact, under $\prob^s$ the  $(\cG^Y_u)_{u\leq t}$-intensity of $H_u$ is given by 
$$\bigg(1 + \frac{\Delta \mu }{\mu_\ell} \ind_{\{u \geq s \}}\bigg)\mu_\ell(1-H_{u^-}) = (\mu_\ell \ind_{\{u < s \}} + \mu_h \ind_{\{u\geq s \}})( 1-H_{u^-})$$
and 
$$\frac{1}{\beta} \bigg(Y_u - \int_0^u \Delta \mu \ind_{\{v\geq s \}}\ud v\bigg)$$ 
is a  $(\cG^Y_u)_{u\leq t}$-Brownian motion. 
Deriving the explicit expression of $Z_t$ and $Z^s_t$, by applying the Dol\'{e}ans-Dade exponential formula (\cite[Corollary 11.5.6, p. 491]{shreve2004stochastic}), we get that 
\begin{equation}
Z_t= \exp \bigg( \int_0^t \frac{\Delta \mu}{\beta} \ud B_u + \frac{1}{2}\int_0^t \frac{(\Delta \mu)^2}{\beta^2} \ud u -\int_0^t \mu_\ell(1-H_{u^-})\ud u + \sum_{u\le t} \ln \Big( 1+ \frac{\Delta \mu}{\mu_\ell} \Delta H_u \Big)\bigg)
\end{equation}
and
\begin{equation}
Z_t^s= \exp \bigg( \int_s^t \frac{\Delta \mu}{\beta} \ud B_u + \frac{1}{2}\int_s^t \frac{(\Delta \mu)^2}{\beta^2} \ud u -\int_s^t \mu_\ell(1-H_{u^-})\ud u + \!\! \sum_{s\le u\le t} \ln \Big( 1+ \frac{\Delta \mu}{\mu_\ell} \Delta H_u \Big)\bigg)
\end{equation}
hence for $t \geq s$,
\begin{equation}\label{C1}
Z^s_t = \frac{Z_t}{Z_s}.
\end{equation}
From  \eqref{C0} and \eqref{C1}, we obtain
\begin{equation}\label{C2}
  \varphi_t = \e^{\lambda t} Z_t \Big ( \frac{\pi}{1 - \pi} + \lambda \int_0^t \frac{\e^{-\lambda s}}{Z_s} \ud s \Big) .
\end{equation}
Recalling that $dZ_t =Z_{t^-} \ud M_t$, with $M$ is defined in \eqref{dM}, we have that
\begin{align}
    \begin{aligned} \label{C2bis}
     d(\e^{\lambda t}Z_t)
     &= \lambda \e^{\lambda t} Z_t \ud t + \e^{\lambda t} Z_{t-}\ud M_t.
    \end{aligned}
\end{align}
Finally, from \eqref{C2} and \eqref{C2bis}, we obtain
\begin{align}
    \begin{aligned}
        \ud \varphi_t & = d(\e^{\lambda t}Z_t) \Big ( \frac{\pi}{1 - \pi} + \lambda \int_0^t \frac{\e^{-\lambda s}}{Z_s} \ud s \Big) + \lambda \ud t\\
        & = (\lambda \e^{\lambda t} Z_t \ud t + \e^{\lambda t} Z_{t-}\ud M_t) \Big ( \frac{\pi}{1 - \pi} + \lambda \int_0^t \frac{\e^{-\lambda s}}{Z_s} \ud s \Big) + \lambda \ud t\\
        &=\lambda( 1 + \varphi_t) \ud t +  \varphi_{t^-} \ud M_t,
    \end{aligned}
\end{align}
which concludes the proof.
\end{proof}

\begin{proof}[\textbf{Proof of Proposition \ref{App:hatBhatm}}]
We start by proving that $(\hat B_t)_{t\ge 0}$ is a $(\bG^Y,\prob_\pi)$-Brownian motion. 
First, we show that $(\hat B_t)_{t\ge 0}$ is a $(\bG^Y,\prob_\pi)$-martingale.
Observe that $(\hat B_t)_{t\ge 0}$ is integrable and $(\bG^Y,\prob_\pi)$-adapted, and, for $t\ge s$
\begin{align}
    \begin{aligned}
        \media_\pi [\hat B_t\given \cG^Y_s]
        &=\frac{1}{\beta} \media_\pi\bigg[ Y_t-Y_s-\Delta \mu \int_s^t \Pi_u \ud u \Given  \cG^Y_s \bigg]+\frac{1}{\beta} \bigg(Y_s-\Delta \mu \int_0^s \Pi_u \ud u\bigg) \\
        &=\frac{1}{\beta} \media_\pi\bigg[ \beta(B_t-B_s)-\Delta \mu \int_s^t (\ind_{\{u\ge \xi\}}-\Pi_u) \ud u \Given  \cG^Y_s \bigg]+ \hat B_s\\
        &=\frac{1}{\beta} \media_\pi[ \beta(B_t-B_s) \given  \cG^Y_s]+ \hat B_s,
    \end{aligned}
\end{align}
where in the third equality we used the definition of $(Y_t)_{t\ge 0}$ \eqref{Eq:dY} and the last equality we used the tower property.
The process $(B_t)_{t\ge 0}$ is, by definition, a $(\bF,\prob_\pi)$-Brownian motion and, using the \textit{Immersion property} (see  Remark \ref{Rm:immersion}),  
it is also a $(\bG,\prob_\pi)$-Brownian motion. By tower property, \[\media_\pi[B_t-B_s\given \cG^Y_s]=\media_\pi[\media_\pi[B_t-B_s\given \cG_s] \given \cG^Y_s]=0.\]

The continuity of $(\hat B_t)_{t\ge 0}$ follows from its definition, and next we show that its quadratic variation is $t$.
From \eqref{Eq:dY}, the quadratic variation of $(Y_t)_{t\ge 0}$ is $\langle Y \rangle_t=\beta^2 t$, yielding $\langle \hat B \rangle_t = \beta^{-2} \langle Y \rangle_t = t.$
From the L\'evy Martingale characterization of the Brownian motion \cite[Theorem 3.16]{karatzas1991brownian}, it follows that $(\hat B_t)_{t\ge 0}$ is a $(\bG^Y,\prob_\pi)$-Brownian motion.

Next, we show that $(\hat m_t)_{t\ge 0}$ is a $(\bG^Y,\prob_\pi)$-martingale. For $t \geq 0$, denote \begin{align}
    M_t \coloneqq  H_t- \int_0^t (1-H_s)\mu_s \ud s,
\end{align} and observe that $(M_t)_{t\ge 0}$ is a $( \bG,\prob_\pi)$-martingale (cf. Remark \ref{Rm:immersion}). Let 
\begin{align}
   \hat M_t \coloneqq  \media_\pi[M_t\given \cG^Y_t] , \quad t \geq 0
\end{align}
Since $\bG^Y\subset \bG$, by tower property, $(\hat M_t)_{t\ge 0}$ is a $(\bG^Y,\prob_\pi)$-martingale. Moreover
\begin{align}
    \begin{aligned}
        \hat M_t&=\media_\pi\bigg[H_t- \int_0^t (1-H_s)\mu_s \ud s \Given \cG^Y_t\bigg]\\
        &=H_t- \int_0^t (1-H_s)\hat \mu_s \ud s + \widetilde M_t=\hat m_t + \widetilde M_t
    \end{aligned}
\end{align}
where $(\widetilde M_t)_{t\ge 0}$ is a $(\bG^Y,\prob_\pi)$-martingale. In the second equality, we used the $\bG^Y_t$-measurability of $H_t$ and \cite[Lemma 8.4]{liptser2013statistics}. This shows that $(\hat m_t)_{t\ge 0}$ is a $(\bG^Y,\prob_\pi)$-martingale and the proof is complete.
\end{proof}

\begin{proof}[\textbf{Proof of Lemma \ref{Lemma4.1}}]

By \textit{Immersion property} (see, Remark \ref{Rm:immersion}) we only need to prove that $(M^\xi_t)_{t\in[0,T]}$ is an $(\bF,\prob_\pi)$-martingale.
Since $\xi$ is independent of $(B_t)_{t\in[0,T]}$ and $\bF = \bF^\mu \vee \bF^B$, using $\bF^\mu$-Markovianity of $(\mu_t )_{t\in[0,T]}$, we have for any $t\geq s \geq 0$ that
\begin{align}
\media_\pi\big[H^\xi_t \mid \cF_s    \big]=\media_\pi\big[H^\xi_t \mid \cF^\mu_s    \big]= \prob_\pi\big( \mu_t=\mu_h\mid \sigma(\mu_s) \big).
\end{align}
Thus, we get 
\begin{align}\media_\pi\big[H^\xi_t |\cF_s    \big]&=\prob_\pi( \mu_t=\mu_h|\mu_s=\mu_\ell) \ind_{\{\mu_s =\mu_\ell\}} + \prob_\pi( \mu_t=\mu_h|\mu_s=\mu_h) \ind_{\{\mu_s =\mu_h\}}\\
&=(1-e^{-\lambda(t-s))})\ind_{\{\xi > s\}} + \ind_{\{\xi \leq s\}}.
\end{align}
Consequently,
\begin{equation}
   \media_\pi\big[H^\xi_t - H^\xi_s|\cF_s    \big]= (1-e^{-\lambda(t-s))})\ind_{\{\xi > s\}}. 
\end{equation}
For any $t\geq s\geq 0$, we also have
\begin{equation}
   \media_\pi\Big[\int_s^t \lambda (1-H^\xi_{u^-})\ud u \Given \cF_s \Big] 
   = \int_s^t \lambda e^{-\lambda(u-s))} \ud u \ind_{\{\xi > s\}}=\media_\pi\big[H^\xi_t - H^\xi_s|\cF_s    \big].
   \end{equation}
   Thus $\media_\pi[M^\xi_t-M^\xi_s \mid \cF_s]=0$, which concludes the proof.
\end{proof}

\begin{proof}[\textbf{Proof of Proposition \ref{MarkovPiH}}]

Let $g(t, x, h)$ be a  measurable function on $(t,x,h) \in [0,T]\times[0,1] \times \{0,1\}$, $C^1$ on $t\in[0,T]$ and $C^2$ on $x\in [0,1]$. Using the filtering equation \eqref{SDEFilterG} and It\^o's formula, we obtain
\begin{align} 
\begin{aligned}\label{ito}
\ud g(t,\Pi_t,H_t) & =  \frac{\partial g}{\partial t}(t,\Pi_t,H_t) \ud t +  \frac{\partial g}{\partial x} (t,\Pi_t, H_t) \ud \Pi^c_t + \frac{1}{2} \frac{\Delta^2 \mu^2}{\beta^2}\Pi^2_t(1-\Pi_t)^2 \frac{\partial^2 g}{\partial x^2}(t,\Pi_t,H_t) \ud t \\  
& \quad  + [ g(t,\Pi_t, H_t) - g(t,\Pi_{t^-}, H_{t^-}) ] \ud H_t,
\end{aligned}
  \end{align} 
where the continuous component of the filter is given by 
\begin{equation}\label{parte continua}
\ud \Pi^c_t =  \lambda(1-\Pi_t)\ud t+\frac{\Delta\mu}{\beta}\Pi_t(1-\Pi_t)\ud \hat{B}_t -   \frac{\Delta\mu (1-\Pi_{t^-})\Pi_{t^-}}  {\mu_\ell + \Delta \mu \Pi_{t^-}}  (1-H_{t^-})(\mu_\ell + \Delta \mu \Pi_{t^-}).
\end{equation}
From \eqref{SDEFilterG}, we get
\begin{align} 
\begin{aligned}\label{salti}
&\big( g(t,\Pi_t, H_t) - g(t,\Pi_{t^-}, H_{t^-}) \big) \ud H_t \\
&= \Big [ g \Big (t,\Pi_{t^-} + \frac {\Delta\mu \Pi_{t^-} (1 - \Pi_{t^-}) }{\mu_\ell + \Delta\mu \Pi_{t^-} }, H_{t^-} +1 \Big ) - g(t,\Pi_{t^-}, H_{t^-}) \Big ] \ud H_t \\
&= \Big [ g \Big ( t,  \frac{\mu_h \Pi_{t^-}}{\mu_\ell + \Delta \mu \Pi_{t^-}}, H_{t^-} + 1 \Big ) - g(t,\Pi_{t^-}, H_{t^-}) \Big ] \ud H_t. 
\end{aligned}
  \end{align}
  Plugging \eqref{parte continua} and \eqref{salti} into \eqref{ito}, we obtain 
  \begin{align} 
\begin{aligned}
\ud g(t,\Pi_t,H_t) &= \mathcal{L}^{(\Pi,H)}g(t,\Pi_t,H_t) \ud t + \frac{\Delta\mu}{\beta}\Pi_t(1-\Pi_t) \frac{\partial g}{\partial x} (t,\Pi_t, H_t) \ud \hat{B}_t  \\ 
&\quad  +\Big [ g \Big ( t,  \frac{\mu_h \Pi_{t^-}}{\mu_\ell + \Delta \mu \Pi_{t^-}}, H_{t^-} +1\Big ) - g(t,\Pi_{t^-}, H_{t^-}) \Big ] \ud \hat m_t,
\end{aligned}
  \end{align}
where the last two terms are $(\bG^Y, \prob_\pi )$-martingales because $g$ and $\frac{\partial g}{\partial x}$ are bounded and this concludes the proof (see, \cite[Proposition 1.7, Chapter IV]{ethier2009markov}).
\end{proof}

\end{appendices}

\bibliography{sn-article}
\end{document}